\documentclass[sigconf]{acmart}

\AtBeginDocument{%
  }

\copyrightyear{2026}
\acmYear{2026}
\setcopyright{cc}
\setcctype{by}
\acmConference[SeQureDB '26]{Workshop on Secure and Private Data Management}{May 31-June 05, 2026}{Bengaluru, India}
\acmBooktitle{Workshop on Secure and Private Data Management (SeQureDB '26), May 31-June 05, 2026, Bengaluru, India}
\setcopyright{acmlicensed}
\acmDOI{10.1145/3807894.3810274}
\acmISBN{979-8-4007-2219-6/2026/05}

\usepackage[normalem]{ulem}

\usepackage{xcolor}
\usepackage{times}
\usepackage{color}
\usepackage{graphicx}
\usepackage{amsmath}
\usepackage{array}
\usepackage{subcaption}
\usepackage{textcomp}
\usepackage{adjustbox}
\usepackage{xspace}
\usepackage{float}
\usepackage{algpseudocode}
\usepackage{wrapfig}
\usepackage{microtype}
\usepackage{multirow}
\usepackage{multicol}
\usepackage{balance}
\usepackage{nowidow}
\usepackage{import}
\usepackage{mathtools}
\usepackage{listings}
\usepackage{url}
\usepackage{hhline}
\usepackage{tabularx}
\usepackage[mathscr]{euscript}
 \let\mathscr\relax
\usepackage{bigints}
\usepackage{adjustbox}
\usepackage{textcomp}
\usepackage{enumitem}
\usepackage{array}
\usepackage{makecell}
\usepackage{bm}
\usepackage{xparse}
\usepackage{booktabs}
\usepackage{colortbl}
\usepackage{algorithm}
\usepackage{algorithmicx}
\usepackage{scalerel}

\DeclareMathOperator*{\argmax}{arg\,max}

\definecolor{lightgray}{gray}{0.92}
\definecolor{okgray}{gray}{0.72}
\definecolor{applegreen}{rgb}{0.13, 0.67, 0.8}
\makeatletter 
\renewcommand{\@seccntformat}[1]{\csname the#1\endcsname\ \ }
\makeatother

\newcommand{\BigO}[1]{\ensuremath{\operatorname{O}\bigl(#1\bigr)}}

\newcommand{\rfig}[1]{Figure~\ref{#1}}

\newcommand{\rtab}[1]{Table~\ref{#1}}
\newcommand{\req}[1]{Eq.~\ref{#1}}
\newcommand{\rsec}[1]{Section~\ref{#1}}

\newcommand{\rlem}[1]{Lemma~\ref{#1}}
\newcommand{\ralg}[1]{Algorithm~\ref{#1}}
\newcommand{\rapp}[1]{Appendix~\ref{#1}}

\newcommand{\mynonparagraph}[1]{\vspace{0.1in} \noindent \textbf{\emph{#1}}\mbox{}}

\newsavebox{\fminipagebox}
\NewDocumentEnvironment{fminipage}{m O{\fboxsep}}
 {\par\kern#2\noindent\begin{lrbox}{\fminipagebox}
  \begin{minipage}{#1}\ignorespaces}
 {\end{minipage}\end{lrbox}%
  \makebox[#1]{%
    \kern\dimexpr-\fboxsep-\fboxrule\relax
    \fbox{\usebox{\fminipagebox}}%
    \kern\dimexpr-\fboxsep-\fboxrule\relax
  }\par\kern#2
 }

\algblock{ParFor}{EndParFor}
\algnewcommand\algorithmicparfor{\textbf{par-for}}
\algnewcommand\algorithmicpardo{\textbf{do}}
\algnewcommand\algorithmicendparfor{\textbf{end\ par-for}}
\algrenewtext{ParFor}[1]{\algorithmicparfor\ #1\ \algorithmicpardo}
\algrenewtext{EndParFor}{\algorithmicendparfor}

\algblock{MyFunc}{EndMyFunc}
\algrenewtext{MyFunc}[3]{#1\ \textbf{#2}\ #3}
\algrenewtext{EndMyFunc}{}
\algrenewcommand\algorithmicindent{1.0em}%
\newcommand{\MyComment}[1]{}
\algdef{SE}[DOWHILE]{Do}{DoWhile}{\algorithmicdo}[1]{\algorithmicwhile\ #1}

\algblock{MyWhile}{EndMyWhile}
\algnewcommand\algorithmicmywhile{\textbf{while}}
\algnewcommand\algorithmicmywhiledo{\textbf{do}}
\algnewcommand\algorithmicendmywhile{\textbf{end while}}
\algrenewtext{MyWhile}[1]{\algorithmicmywhile\ #1\ \algorithmicpardo}
\algrenewtext{EndMyWhile}{\algorithmicendmywhile}

\theoremstyle{plain}

\theoremstyle{definition}

\algnewcommand\algorithmicforeach{\textbf{for each}}
\algdef{S}[FOR]{ForEach}[1]{\algorithmicforeach\ #1\ \algorithmicdo}

\newcommand{\etal}{\hbox{\emph{et al.}}\xspace}
\newcommand{\eg}{\hbox{\emph{e.g.}}\xspace}
\newcommand{\ie}{\hbox{\emph{i.e.}}\xspace}

\newcommand{\eedp}{\ensuremath{\varepsilon\text{-edge-DP}}}

\usepackage{amsmath}
\usepackage{amsthm}

\newtheorem{lemma}{Lemma}[section]
\theoremstyle{definition}

\theoremstyle{remark}

\begin{document}

\title{Estimating Power-Law Exponent \\ with Edge Differential Privacy}
\titlenote{This version adds an appendix to the published paper.}

\author{Adam Tan}
\email{wat@sfu.ca}
\affiliation{%
  \institution{Simon Fraser University}
  \city{Burnaby}
  \state{BC}
  \country{Canada}
}

\author{Mohamed Hefny}
\email{mohamed_hefny@sfu.ca}
\affiliation{%
  \institution{Simon Fraser University}
  \city{Burnaby}
  \state{BC}
  \country{Canada}
}

\author{Keval Vora}
\email{keval@cs.sfu.ca}
\affiliation{%
  \institution{Simon Fraser University}
  \city{Burnaby}
  \state{BC}
  \country{Canada}
}

\begin{abstract}
Many real-world graphs have degree distributions that are well
approximated by a power-law, and the corresponding scaling
parameter $\alpha$ provides a compact summary of that structure which is useful for graph analysis and system optimization.
When graphs contain sensitive relationship data,
$\alpha$ must be estimated without revealing information about
individual edges. This paper studies power-law exponent estimation
under edge differential privacy. Instead of first releasing a noisy
degree distribution and then fitting a power-law model, we propose
privatizing only the low-dimensional sufficient statistics needed to
estimate $\alpha$, thereby avoiding the high distortion introduced by
traditional approaches. 
Using these released statistics, we support both discrete
approximation and likelihood-based numerical optimization
for efficient parameter estimation. 
We develop edge-DP algorithms for both centralized and local DP models, 
compare degree release and log-statistic release in the local
setting, and evaluate the resulting methods on various graph datasets
across multiple privacy budgets and tail-cutoff settings.
\end{abstract}

\maketitle

\section{Introduction}

Graph databases are useful in domains where relationships are central to the data, enabling efficient
structure-based information retrieval. 
However, graphs often contain sensitive information, and there is a need to develop privacy-preserving graph analysis techniques that prevent sensitive information from being leaked.

Among the structural properties studied in graphs, degree distributions are especially important. Many real-world graphs exhibit scale-free behavior, in which a small number of nodes act as highly connected hubs while most nodes have relatively few connections. This pattern appears across domains such as the web, social networks, and online retail, among many others~\cite{Newman2005Power}.

Such behavior is often modeled by a \emph{power-law distribution}~\cite{mitzenmacher2004brief,Newman2005Power,clauset2009powerlaws}, typically expressed as $P(d) \propto d^{-\alpha}$, meaning that the probability $P(d)$ that a node has degree $d$ decreases polynomially, with $\alpha$ being a constant parameter of the distribution known as the \emph{scaling parameter}, also referred to as the power-law exponent. This means high-degree nodes are rare while low-degree nodes are much more common.

Estimating the scaling parameter $\alpha$ of a power-law distribution helps tailor graph algorithms and systems that rely on degree information~\cite{li2017experimental,fastexacthub,jiang2014hop,chung2002connected,tang2015optimizing,vora2019lumos}. Practitioners typically estimate $\alpha$ by maximum likelihood, either through a closed-form discrete approximation or through numerical optimization~\cite{clauset2009powerlaws}. For graphs containing sensitive relationship data, however, we must perform this estimation under privacy constraints so that the released parameter does not reveal sensitive information about individual edges or graph structure.

To estimate $\alpha$ while protecting sensitive graph information, we use differential privacy (DP)~\cite{dpbook}, which provides the formal framework for this goal by enabling data analysis while protecting sensitive graph information. Existing DP methods for graph analysis~\cite{ldpstats, shuntriangles, Dhulipala22, sectric, privgraph,privagm} provide privacy guarantees for structural information in graphs. However, they do not study private estimation of the power-law scaling parameter $\alpha$ directly.

To address this gap, we develop algorithms for estimating the power-law scaling parameter $\alpha$ under \emph{edge differential privacy} \mbox{(edge-DP)}, which protects individual edges.

A common baseline for private $\alpha$ estimation, used for example by Hay et al.~\cite{degreedistribution}, is to first release a DP degree-distribution histogram, that is, counts of how many nodes have degree $0,1,2,\ldots$, and then fit a power-law model to the privatized histogram via MLE. However, when the goal is to estimate a single scalar parameter, this pipeline of releasing a histogram and then fitting a model is inefficient.
The noise added to each degree count, together with smoothing, binning, and projection steps, can distort the tail of the degree distribution before fitting, which leads to inaccurate $\alpha$ estimates with high variance.

Instead, we privatize only the low-dimensional statistics needed for the $\alpha$ estimation. We start from the discrete approximation estimator for the power-law scaling parameter~\cite{clauset2009powerlaws}, decompose it into low-sensitivity sub-components, and apply the Laplace mechanism to each component before recombining them into a private estimate. We then show how the same privatized statistics can also support maximum likelihood estimation via numerical optimization, allowing us to obtain both discrete-approximation and numerical-optimization variants under edge-DP. Because the
released quantities have low sensitivity, the Laplace mechanism adds
relatively small noise.

We develop differentially private algorithms under both the centralized and local
models. In the centralized model, a trusted curator has access to the
entire graph and releases noisy estimates of the required low-dimensional statistics.
In the local model, there is no trusted curator, and each node
perturbs its own edge-related statistics before release. We study two
local release strategies: degree release and log-statistic release.

Together, these choices give two centralized edge-DP algorithms, one based on discrete approximation and another based on numerical optimization,
and four local edge-DP variants obtained by combining degree
release or log-statistic release with discrete approximation or
numerical optimization.

We evaluate the accuracy of these methods on 6 publicly
available graph datasets and 3 synthetic datasets. Our results show
that directly privatizing the sufficient statistics needed to estimate
$\alpha$ is more accurate and more stable than histogram-based
fitting in the centralized model. 
Among the methods that privatize the
sufficient statistics directly, numerical optimization is overall
more accurate than discrete approximation.

\section{Background and Setup}
This section introduces the graph model, the power-law estimation
setup, and the privacy definitions used throughout the paper.
Table~\ref{tab:notations} summarizes the main notations.

\begin{table}[t]
    \centering
    \footnotesize
    \caption{Notations.}
    \label{tab:notations}
    \vspace{-0.1in}
    \begin{tabular}{c|l} 
         Symbols &  \multicolumn{1}{c}{Description} \\\midrule
         $G = (V, E)$& Graph with nodes $V$ and edges $E$  \\
         $d_{v}$&  Degree of node $v \in V$ \\
         $d_{\min}, d_{\max}$&  Degree range of MLE fit $(d_{\min}, d_{\max})$ \\
         $D'$& Degrees in $G$ between $d_{\min}$ and $d_{\max}$ \\
         $T_{disc}, N$& Statistics for $\alpha$ estimation \\
         $\tilde{d}_{v}, \tilde{T}_{disc}, \tilde{N}$ & DP estimates of $d_{v}, T_{disc}$ and $N$ \\
         $\hat{\alpha}_{\textsc{central}}, \hat{\alpha}_{\textsc{local}}$& Centralized DP and local DP $\alpha$ estimates \\ 
    \end{tabular}
    \vspace{0.1in}
\end{table}

\subsection{Power-Law Degree Distribution}
\vspace{-0.05in}
\mynonparagraph{Graph.} We consider a simple undirected graph $G = (V,E)$, where $V$ denotes the set of nodes and $E$ the set of edges. The degree of a node $v \in V$ is denoted by $d_v$. 

\vspace{-0.03in}
\mynonparagraph{Power-Law Degree Distribution.} The degree distribution of a graph $G$ is the probability distribution $P(d)$ that a randomly selected node $v \in V$ has degree $d_v = d$. Many real-world graphs are scale-free, with a small number of highly connected nodes and many low-degree nodes. A power-law degree distribution models scale-free graphs by $P(d) \propto d^{-\alpha}$ for $d \geq d_{\min}$, indicating the coexistence of a few highly connected nodes and many sparsely connected ones. Nodes with degrees $d_v \geq d_{\min}$ are called tail nodes. 

The scaling parameter $\alpha$ is a single parameter that governs the heaviness of the distribution's right tail. It is typically below 3, with occasional exceptions~\cite{clauset2009powerlaws}. Smaller values of $\alpha$ correspond to heavier tails and a higher likelihood of extreme high-degree nodes, whereas larger values imply a more rapid decay and consequently a more homogeneous connectivity structure within the network. 

\vspace{-0.03in}
\mynonparagraph{Discrete Power-Law Distribution.}
Since node degrees are integers, the degree distribution of tail nodes follows a discrete power-law distribution with parameter $\alpha$.
The node degrees in the tail are independent and identically distributed according to a truncated discrete power-law distribution with probability mass function $P(d \mid \alpha)$:
\begin{equation}
  P(d \mid \alpha) = \frac{d^{-\alpha}}{Z(\alpha)} 
  \qquad
  Z(\alpha) = \sum_{d = d_{\min}}^{d_{\max}} d^{-\alpha} \label{eq:discretepowerlaw}
\end{equation}
for $d \in \{d_{\min}, \dots, d_{\max}\}$ where $d_{\max}$ represents a known upper bound on the node degrees and $Z(\alpha)$ is the normalizing constant based on Hurwitz-$\zeta$-function~\cite{bauke2007powerlaws}.
The tail degrees with power-law distribution are captured in the multiset $D' = [d_v : v \in V \ \land \ d_{\min} \leq d_v \leq d_{\max}]$, and $N = |D'|$ is the number of degrees in the multiset $D'$. 

\mynonparagraph{Maximum Likelihood Estimation for $\alpha$.}
A common way to estimate $\alpha$ is through Maximum Likelihood Estimation (MLE). Clauset~\etal~\cite{clauset2009powerlaws} showed that a closed-form discrete approximation estimator for observed degrees $d_i \in D'$ can be defined as follows:
\begin{equation}
  \hat{\alpha}_{\mathrm{disc}}
  = 1 + \frac{N}{T_{\text{disc}}}
  \qquad
  T_{\text{disc}} = \sum_{i=1}^N \ln\!\left(\frac{d_i}{d_{\min} - 0.5}\right) 
  \label{eq:discapproximation}
\end{equation}
This discrete approximation provides a closed-form estimate which may generally be good enough for most practical purposes. Alternatively,  $\hat{\alpha}_{\mathrm{disc}}$ can be obtained by numerical optimization as described next. 

The likelihood function for the discrete power-law model in \req{eq:discretepowerlaw} with $d_i \in D'$ is~\cite{bauke2007powerlaws,clauset2009powerlaws}:
\begin{equation*}
    \ell(\alpha; D')
  = \sum_{i=1}^N \log P(d_i \mid \alpha)
  = -\alpha \sum_{i=1}^N \ln d_i - N \ln Z(\alpha)
\end{equation*}
With the aggregated statistic $T = \sum_{i=1}^N \ln d_i$, we can express the log-likelihood in terms of $T$ as follows:
\begin{equation}
    \ell(\alpha; T,N)=-\alpha T-N \ln Z(\alpha)
\end{equation}
where the sufficient statistics for $\alpha$ are the pair $(T, N)$, since $Z(\alpha)$ does not depend on observed $d_i$. Therefore, the maximum likelihood estimator is: 
\begin{equation}
  \hat{\alpha}_{\mathrm{disc}}
  = \argmax_{\alpha > 0} \ell(\alpha; T, N)
  = \argmax_{\alpha > 0} \bigl[-\alpha T - N \ln Z(\alpha)\bigr]
  \label{eq:discloglikelihood}
\end{equation}

\vspace{0.02in}
\subsection{Privacy Model}
In graph data where sensitive information lies in connections between entities, edge differential privacy (edge-DP)~\cite{edgedporiginal} ensures analyses do not reveal individual edges. Similar to the original DP formulation by Dwork \textit{et al.} \cite{dpbook}, with edge-DP each edge is treated as an individual entry in a database (or graph) $G$.

\vspace{-0.01in}
\mynonparagraph{Edge Differentially Private $\alpha$ Estimation.}
A randomized $\alpha$ estimation algorithm $\mathcal{A}(G)$ that takes input graph $G$ 
and outputs some $\hat{\alpha}_{\mathrm{disc}}$ value from output space $\mathcal{R}$ is \emph{$\varepsilon$ edge differentially private} (\eedp) if for all $R \subseteq \mathcal{R}$, and neighboring graphs $G$ and $G'$,
\begin{gather*}
    Pr(\mathcal{A}(G) \in R)\leq e^{\varepsilon}Pr(\mathcal{A}(G') \in R)
\end{gather*}

Here, neighboring graphs $G$ and $G'$ share the same set of nodes but differ in one edge (\ie, the size of the symmetric difference of their edge sets is 1). 
The $\varepsilon$ is referred to as the \emph{privacy budget} as it governs the amount of random noise added to the $\hat{\alpha}_{\mathrm{disc}}$ value.

\vspace{-0.01in}
\mynonparagraph{Centralized and Local Edge-DP Models.}
Edge differential privacy can be realized under two models based on data visibility: the central model and the local model.

In the centralized model, a trusted curator maintains the entire graph and applies a randomized algorithm to ensure that the presence or absence of any single edge cannot be inferred. This is suitable for traditional database scenarios like a curator-managed social network, where the entire graph is safely accessible. 

However, a trusted curator of data store having access to the entire graph can become impractical in modern systems that rely on decentralized or federated architectures. In local edge differential privacy (LEDP)~\cite{ldpstats,ldpsynth,  shuntriangles, Dhulipala22, centrallocaledp}, each node retains ground truth to their associated data, and aggregations on the graphs are constructed using $\eedp$ queries to the nodes. Hence, the $\hat{\alpha}_{\mathrm{disc}}$ computed under LEDP is based on the degree estimates that must be computed from individually perturbed edges. 

\subsection{Our Goal}
Our goal is to design \eedp{} algorithms that estimate the power-law
scaling parameter of a graph $G$ over its fitted tail $D'$ while
preserving high utility under privacy constraints. In
the centralized model, the algorithm accesses the entire graph and
produces a private estimate $\hat{\alpha}_{\textsc{central}}$. In the
LEDP model, each node ensures \eedp{} locally on its own
edge-related information, and the reports are aggregated to produce private estimate $\hat{\alpha}_{\textsc{local}}$.

\section{Centralized Algorithms}
We develop two DP algorithms in centralized model to estimate scaling parameter $\hat{\alpha}_{\textsc{central}}$: one using discrete approximation and other via numerical optimization. 
Both algorithms first compute noisy statistics for $T_{\text{disc}}$ and $N$ as defined in \req{eq:discapproximation} using Laplace mechanism.
And then, we use these to compute $\hat{\alpha}_{\textsc{central}}$ using two approaches (\rsec{sec:centralda} and \rsec{sec:centralno}). 

To compute the noisy statistics using Laplace mechanism, we first analyze their sensitivities as described next.

\subsection{Sensitivity Analysis}
\begin{lemma}
\label{lma:tdsicsensitivity}
    Global sensitivity of $T_{\text{disc}}$ is bounded by \BigO{
\ln (
    \frac{
      d_{\min} + 1
    }{
      d_{\min}
    }
  )
 }.
\end{lemma}
\begin{proof}
Consider neighboring graphs $G$ and $G'$ that differ in exactly one edge $(u, w)$. Only degrees of $u$ and $w$ change across these two graphs. Let $d_u'$ and $d_w'$ be respectively the degrees of nodes $u$ and $w$ in $G'$. Without loss of generality, $d_u - d_u' = 1$ and $d_w - d_w' = 1$.
Define the per-node contribution $t_v(G)$ as:
\[
  t_v(G)
  =
  \begin{cases}
    \ln\!\left(
    \frac{
      d_v(G)
    }{
      d_{\min} - 0.5
    }
  \right)
    & \quad d_v(G) \ge d_{\min}\\
    0 & \quad d_v(G) < d_{\min}
  \end{cases}
\]
The change in a node's contribution will depend on whether its degree remains in the tail. If it does (\ie, $d_v \geq d_{\min}$ and $d_v' \geq d_{\min}$), then the change in the contribution is:
\[
\begin{aligned}
  t_v(G) - t_v(G')
  & =
  \ln\!\left(
    \frac{
      d_v
    }{
      d_{\min} - 0.5
    }
  \right) 
   -
  \ln\!\left(
    \frac{
      d_v'
    }{
      d_{\min} - 0.5
    }
  \right) \\
  & =
  \ln\!\left(
    \frac{
      d_v
    }{
      d_v'
    }
  \right) 
   =
  \ln\!\left(
    \frac{
      d + 1
    }{ d
    }
  \right)
\end{aligned}
\]
where $d$ is an arbitrary integer such that $d \geq d_{\min}$. Because this expression is strictly decreasing in $d,$ the largest possible difference will occur when $d = d_{\min}$:
\[
\begin{aligned}
  \left\lvert t_v(G) - t_v(G')\right\rvert
  & \leq
  \ln\!\left(
    \frac{
      d_{\min} + 1
    }{
      d_{\min}
    }
  \right)
\end{aligned}
\]
When $d_{\min} = 1$, this difference becomes $\ln 2$. The difference decreases as $d_{\min}$ grows, because $\frac{d_{\min}+1}{d_{\min}}$ decreases as $d_{\min}$ increases.

When a node's degree crosses from being below $d_{\min}$ to above $d_{\min}$ (or vice versa), then its contribution changes from $0$ to
$
  \ln\!\left(
    \frac{
      d_{\min}
    }{
      d_{\min} - 0.5
    }
  \right).
$
The absolute value of this quantity is also no larger than 
$
  \ln\!\left(
    \frac{
      d_{\min} + 1
    }{
      d_{\min}
    }
  \right)
$
when $d_{\min} \geq 1$. 

Therefore, the sensitivity of $t_v(G)$ is no larger than 
$
  \ln\!\left(
    \frac{
      d_{\min} + 1
    }{
      d_{\min}
    }
  \right)
$
for any node, and since a single edge affects at most two nodes, the total sensitivity of
$
  T_{\text{disc}} = \sum_v t_v(G)
$
is bounded by:
\[
\begin{aligned}
  \Delta T_{\text{disc}}
  & =
  \max_{\text{neighbors } G, G'}
    \left\lvert T_{\text{disc}}(G) - T_{\text{disc}}(G') \right\rvert
   \leq
  2 \ln\!\left(
    \frac{
      d_{\min} + 1
    }{
      d_{\min}
    }
  \right)
\end{aligned}
\]
\end{proof}

The above $d_{\min}$ dependent bound on the global sensitivity of
$T_{\text{disc}}$ is important because it remains small for all
relevant choices of $d_{\min}$. Even at the smallest value,
$d_{\min}=1$, the sensitivity is only $2\ln 2 \approx 1.386$. As
$d_{\min}$ increases, this bound decreases, so for a fixed privacy
budget the released value $\tilde{T}_{\text{disc}}$ stays closer to the
true statistic. The overall effect of $d_{\min}$ on final estimation
accuracy, however, depends on additional factors and is evaluated in
Section~\ref{sec:evaluation}.

\begin{lemma}
Global sensitivity of $N$ is at most 2.
\end{lemma}
\begin{proof}
Consider neighboring graphs $G$ and $G'$ that differ in exactly one edge. The only nodes that can have their tail membership changed are the endpoints of the edge. Each endpoint can either enter the tail (if previously out of the tail) or exit the tail (if previously in the tail). Since only two nodes are affected, the total change in the number of nodes in the tail is:
\[
  \Delta N
  =
  \max_{\text{neighbors } G, G'}
    \left\lvert N(G) - N(G') \right\rvert 
  \ \leq \ 2
\]
\end{proof}

\begin{algorithm}[t]
\small
    \caption{DP $\hat{\alpha}_{\textsc{central}}$ via Discrete Approximation}
    \label{alg:DEDP}
    \begin{algorithmic}[1]
        \Require{Graph $G$, Minimum Degree $d_{\min}$, Maximum Degree $d_{\max}$, \newline Privacy Budget $\varepsilon = \varepsilon_t + \varepsilon_n$}
        \Ensure{DP Estimate of $\alpha$}
        \State $D' = [d_v: \forall v \in G.V \land d_{\min} \leq d_v \leq d_{\max}]$
        \State $T_{disc}\gets \sum\limits_{d \in D'} \ln\!\left(\frac{d}{d_{\min}-0.5}\right)$
        \State $\tilde{T}_{disc} \gets T_{disc} + \textsc{Lap}(\frac{2 \times ln(\frac{d_{\min}+1}{d_{\min}})}{\varepsilon_t})$
        \State $\tilde{N} \gets |D'| + \textsc{Lap}(\frac{2}{\varepsilon_n})$
        \State \Return $1 + \frac{\tilde{N}}{\tilde{T}_{disc}}$
    \end{algorithmic}
\end{algorithm}

\subsection{$\hat{\alpha}_{\textsc{central}}$ via Discrete Approximation} 
\label{sec:centralda}
Using global sensitivities $\Delta T_{\text{disc}}$ and $\Delta N$, the DP estimate $\hat{\alpha}_{\textsc{central}}$ is computed using the Laplace mechanism~\cite{dpbook}. 

\ralg{alg:DEDP} computes $\hat{\alpha}_{\textsc{central}}$ using the discrete approximation estimator from \req{eq:discapproximation}.
Lines 1-2 compute $D'$ and $T_{\text{disc}}$ using node degrees. Lines 3-4 add Laplace noise proportional to the global sensitivities to compute noisy $\tilde{T}_{\text{disc}}$ and $\tilde{N}$. The $\varepsilon$ budget is split into $\varepsilon_t$ and $\varepsilon_n$ while adding Laplace noise for $\tilde{T}_{\text{disc}}$ and $\tilde{N}$ respectively. Finally, line 5 inserts the noisy estimates into \req{eq:discapproximation} to obtain the DP estimate $\hat{\alpha}_{\textsc{central}}$; by post-processing, this step has no privacy loss.

\begin{lemma}
Using the Laplace Mechanism and Sequential Composition~\cite{dpbook}, $\hat{\alpha}_{\textsc{central}}$ computed by \ralg{alg:DEDP} is $(\varepsilon_t + \varepsilon_n)$-edge differentially private. 
\end{lemma}

\subsection{$\hat{\alpha}_{\textsc{central}}$ via Numerical Optimization}
\label{sec:centralno}

Instead of the above closed-form estimation, we can estimate $\hat{\alpha}_{\textsc{central}}$ by numerically optimizing the discrete log-likelihood. The key idea is to reuse the same noisy $\tilde{T}_{\text{disc}}$ and $\tilde{N}$ estimates, as described next. 

From the definition of $T_{\text{disc}}$ in \req{eq:discapproximation}:
\[
\begin{aligned}
T_{\text{disc}}&=\sum_{i=1}^{N}\ln\left(\frac{d_i}{d_{\min}-0.5}\right) 
 =\sum_{i=1}^{N}\ln d_i-N\ln(d_{\min}-0.5)
\end{aligned}
\]
Therefore, 
\vspace{-0.1in}
\[
\begin{aligned}
\sum_{i=1}^{N}\ln d_i&=T_{\text{disc}}+N\ln(d_{\min}-0.5)
\end{aligned}
\]

Hence, we can compute the DP estimate of this sum using the noisy statistics $\tilde{T}_{\text{disc}}$ and $\tilde{N}$:
\begin{equation}
\label{eq:numanalysislnd}
\begin{aligned}
\widetilde{\sum\ln d_i}&=\tilde{T}_{\text{disc}}+\tilde{N}\ln(d_{\min}-0.5)
\end{aligned}
\end{equation}

Hence, based on \req{eq:discloglikelihood}, our central DP MLE is:

\[
\begin{aligned}
\hat{\alpha}_{\textsc{central}}&=\argmax_{\alpha > 0} 
 \Bigl[-\alpha\widetilde{\sum \ln d_i}-\tilde{N}\ln Z(\alpha)\Bigr]
\end{aligned}
\]

\setlength{\textfloatsep}{0pt}
\begin{algorithm}[t]
\small
    \caption{DP $\hat{\alpha}_{\textsc{central}}$ via Numerical Optimization}
    \label{alg:mledp}
    \begin{algorithmic}[1]
        \Require{Graph $G$, Minimum Degree $d_{\min}$, Maximum Degree $d_{\max}$, \newline Privacy Budget $\varepsilon = \varepsilon_t + \varepsilon_n$}
        \Ensure{DP Estimate of $\alpha$}
        \State $D' = [d_v: \forall v \in G.V \land d_{\min} \leq d_v \leq d_{\max}]$
        \State $T_{disc} \gets \sum\limits_{d \in D'} \ln\!\left(\frac{d}{d_{\min}-0.5}\right)$
        \State $\tilde{T}_{disc} \gets T_{disc}  + \textsc{Lap}(\frac{2 \times ln(\frac{d_{\min}+1}{d_{\min}})}{\varepsilon_t})$
        \State $\tilde{N} \gets |D'| + \textsc{Lap}(\frac{2}{\varepsilon_n})$
        \State \Return$\underset{\alpha > 0}{\mathrm{argmax}}\, \textsc{LogLikelihood}(\tilde{T}_{disc}, \tilde{N}, d_{\min}, d_{\max}, \alpha)$ \Comment{Algo. \ref{alg:loglikelihood}}
    \end{algorithmic}
\end{algorithm}

\setlength{\textfloatsep}{8pt}
\begin{algorithm}[t]
\small
    \caption{Log-likelihood score}
    \label{alg:loglikelihood}
    \begin{algorithmic}[1]
        \Require{$T_{\text{disc}}$, $N$, Minimum Degree $d_{\min}$, Maximum Degree $d_{\max}$, Alpha $\alpha$}
        \Ensure{Log-Likelihood score for $\alpha$}
        \State $Z \gets \sum\limits_{d = d_{\min}}^{d_{\max}} d^{-\alpha}$
        \If {$Z \le 0$}
            \State \Return $- \infty$
        \EndIf
        \State $S \gets T_{\text{disc}} + N\times ln(d_{\min}-0.5)$
        \State \Return $-\alpha \times S - N \times ln(Z)$
    \end{algorithmic}
\end{algorithm}

\ralg{alg:mledp} computes $\hat{\alpha}_{\textsc{central}}$ by maximizing this DP log-likelihood objective. The computation of noisy statistics $\tilde{T}_{\text{disc}}$ and $\tilde{N}$ in lines 1-4 is same as in that in the previous algorithm, using Laplace noise with split budgets. In our experiments, we set
$\varepsilon_t=\varepsilon_n=\varepsilon/2$ for simplicity.
Using these DP estimates, MLE is numerically computed as post-processing step on line 5, with the log-likelihood computation shown in \ralg{alg:loglikelihood}.

\begin{lemma}
Using the Laplace Mechanism and Sequential Composition~\cite{dpbook}, $\hat{\alpha}_{\textsc{central}}$ computed by \ralg{alg:mledp} is $(\varepsilon_t + \varepsilon_n)$-edge differentially private. 
\end{lemma}

\vspace{-0.1in}
\mynonparagraph{Discussion.}
The normalization constant $Z(\alpha)$ needs to know $d_{\max}$. In practice, if the maximum degree of the graph is known or can be assumed to be public knowledge (\eg, maximum degree in social graphs is often visible or reported), it can directly be used without adding noise. On the other hand, $d_{\max}$ is just a single scalar number, so its private estimation requires much less $\varepsilon$ budget than that for the entire degree distribution. Furthermore for power-law distributions, setting $d_{\max}$ conservatively high enough (\eg, $d_{\max}=|V|$) has very little impact on the MLE with no privacy cost. 

\section{Local Algorithms}
In the local model, each node perturbs its own information before releasing. Hence, the aggregator would never see the raw node degrees, and instead operate on noisy per-node statistics produced by local DP mechanisms.

We explore two approaches in this model, both with Laplace mechanism. The first approach computes LEDP degrees and uses them for $\hat{\alpha}_{\textsc{local}}$ estimation. The second approach is consistent with the central model; here, each node releases the noisy log-function statistic required to compute $\tilde{T}_{\text{disc}}$. These approaches result in four algorithms depending on the use of closed form discrete approximation versus numerical optimization using noisy statistics. 

\vspace{-0.03in}
\mynonparagraph{Approach 1: Release Degree Statistic.}
LEDP degree is computed with each node releasing its noisy degree~\cite{ldpsynth}. Hence, the contribution from each node $v$ is simply its DP degree estimate $\tilde{d_v}$.
Hence, $\tilde{T}_{\text{disc}}$ and $\tilde{N}$ are defined as:
\begin{equation}
\label{eq:degreestat}
\begin{aligned}
  \tilde{T}_{\text{disc}}   = \sum_{\tilde{d_v} \geq d_{\min}} \ln\!\left(\frac{\tilde{d_v}}{d_{\min}-0.5}\right) \qquad
  \tilde{N} =  \sum_{\tilde{d_v} \geq d_{\min}}  1 \\
\end{aligned}
\end{equation}

\vspace{-0.03in}
\mynonparagraph{Approach 2: Release Log Statistic.}
Here, the contribution from each node $v$ is modeled as $\tilde{c_v}$:
\begin{equation}
\label{eq:logstat}
  \tilde{c_v} =
  \begin{cases}
    \widetilde{\ln} \left(\dfrac{d_v}{d_{\min}-0.5}\right) & \quad \tilde{d_v} \geq d_{\min}\\[4pt]
    0 & \quad \tilde{d_v} < d_{\min}
  \end{cases}
\end{equation}
where $\tilde{d_v}$ is the DP estimate of $d_v$, and $\widetilde{\ln}\!\left(\frac{d_v}{d_{\min}-0.5}\right)$ denotes the DP estimate of $\ln\!\left(\frac{d_v}{d_{\min}-0.5}\right)$. Hence, $\tilde{T}_{\text{disc}}$ and $\tilde{N}$ are defined as:
\begin{equation}
\label{eq:logstatfinal}
\begin{aligned}
  \tilde{T}_{\text{disc}} = \sum_{v} \tilde{c_v} \qquad
  \tilde{N} = \hspace{-0.2in} \sum_{\tilde{c_v} \ \geq \ \ln \big(\frac{d_{\min}}{d_{\min}-0.5}\big)} \hspace{-0.2in} 1
\end{aligned}
\end{equation}

\subsection{Sensitivity Analysis}
We analyze the global sensitivity of the degree statistic and log statistic to guide the Laplace noise addition. 

\begin{lemma}
Global sensitivity of node degree is 1. 
\end{lemma}
\begin{proof}
    Adding or removing one edge can only change the degree of the endpoints of that edge by 1.
\end{proof}

\begin{lemma}
    Global sensitivity of log statistic from \req{eq:logstat} is at most $
\ln (
    \frac{
      d_{\min} + 1
    }{
      d_{\min}
    }
  )
 $. 
\end{lemma}
\begin{proof}
The proof follows a similar argument to that for \rlem{lma:tdsicsensitivity} proof.
Consider neighboring graphs $G$ and $G'$ that differ in exactly one edge $(u, w)$. Only degrees of $u$ and $w$ change across these two graphs. Let $d_u'$ and $d_w'$ be respectively the degrees of nodes $u$ and $w$ in $G'$. Without loss of generality, $d_u - d_u' = 1$ and $d_w - d_w' = 1$.

The change in a node's contribution $c_v$ will depend on whether its degree remains in the tail. If it does (\ie, $d_v \geq d_{\min}$ and $d_v' \geq d_{\min}$), then the change in the contribution is:
\[
\begin{aligned}
  c_v(G) - c_v(G')
  & =
  \ln\!\left(
    \frac{
      d_v
    }{
      d_{\min} - 0.5
    }
  \right) 
   -
  \ln\!\left(
    \frac{
      d_v'
    }{
      d_{\min} - 0.5
    }
  \right) \\
  & =
  \ln\!\left(
    \frac{
      d_v
    }{
      d_v'
    }
  \right) 
   =
  \ln\!\left(
    \frac{
      d + 1
    }{ d
    }
  \right)
\end{aligned}
\]
where $d$ is an arbitrary integer such that $d \geq d_{\min}$. Because this expression is strictly decreasing in $d,$ the largest possible difference will occur when $d = d_{\min}$:
\[
\begin{aligned}
  \left\lvert c_v(G) - c_v(G')\right\rvert
  & \leq
  \ln\!\left(
    \frac{
      d_{\min} + 1
    }{
      d_{\min}
    }
  \right)
\end{aligned}
\]
When $d_{\min} = 1$, this difference becomes $\ln 2$. The difference decreases as $d_{\min}$ grows, because the logarithm function is monotonically increasing. 

When a node's degree crosses from being below $d_{\min}$ to above $d_{\min}$ (or vice versa), then its contribution changes from $0$ to
$
  \ln\!\left(
    \frac{
      d_{\min}
    }{
      d_{\min} - 0.5
    }
  \right).
$
The absolute value of this quantity is also no larger than 
$
  \ln\!\left(
    \frac{
      d_{\min} + 1
    }{
      d_{\min}
    }
  \right)
$
when $d_{\min} \geq 1$. 

Thus in all cases the sensitivity of $c_v$ is at most 
$
  \ln\!\left(
    \frac{
      d_{\min} + 1
    }{
      d_{\min}
    }
  \right)
$.
\end{proof}

\vspace{-0.035in}
The global sensitivity of $c_v$ benefits from the similar \mbox{$d_{\min}$-dependent} bound as in the central model.

\setlength{\textfloatsep}{4pt}
\begin{algorithm}[t]
\small
    \caption{DP $\hat{\alpha}_{\textsc{local}}$ via Degree Release}
    \label{alg:mleldpdegreestat}
    \begin{algorithmic}[1]
        \Require{Graph $G$, Minimum Degree $d_{\min}$, Maximum Degree $d_{\max}$, \newline Privacy Budget $\varepsilon$}
        \Ensure{DP Estimate of $\alpha$}
        \Statex \texttt{/* Step 1: Release local degree statistic */}
        \ForEach{$v \in G.V$}     
                \State $\tilde{d_v} \gets d_v + Lap(\frac{2}{\varepsilon})$
                \Comment{$\frac{\varepsilon}{2}$ budget split}
                \State \textsc{Release} $\tilde{d_v}$
        \EndFor
        \Statex 
        \Statex \texttt{/* Step 2: Aggregate local statistics */}
        \State $\tilde{T}_{\text{disc}} \gets 0$; $\tilde{N} \gets 0$
        \For{$v \in G.V$} 
        \State $\tilde{d} \gets \textsc{DegreeStatistic}(v)$
        \If{$\tilde{d} \geq d_{\min}$}
            \State $\tilde{T}_{\text{disc}} \gets \tilde{T}_{\text{disc}} + \ln(\frac{\tilde{d}}{d_{\min}-0.5})$
            \State $\tilde{N} \gets \tilde{N} + 1$
        \EndIf
        \EndFor
        \Statex 
        \Statex \texttt{/* Step 3: Estimate $\alpha$ */}
        \Statex \texttt{/* Option A: Discrete Approximation */}
        \State \quad $\hat{\alpha}_{\textsc{local}} \gets 1 + \frac{\tilde{N}}{\tilde{T}_{\text{disc}}}$
        \Statex \texttt{/* Option B: Numerical Optimization */}
        \State \quad $\hat{\alpha}_{\textsc{local}} \gets \underset{\alpha > 0}{\mathrm{argmax}}\, \textsc{LogLikelihood}(\tilde{T}_{\text{disc}}, \tilde{N}, d_{\min}, d_{\max}, \alpha)$ 
        \Statex
        \State \Return $\hat{\alpha}_{\textsc{local}}$ 
    \end{algorithmic}
\end{algorithm}

\vspace{-0.025in}
\subsection{$\hat{\alpha}_{\textsc{local}}$ via Degree Release}
\ralg{alg:mleldpdegreestat} shows the LEDP computation for $\hat{\alpha}_{\textsc{local}}$ using Laplace mechanism for degree release. 
In the first step (lines 1-4), each node releases its DP degree estimate $\tilde{d_v}$ computed using Laplace noise proportional to sensitivity 1. The $\varepsilon$ privacy budget is divided by 2 as each edge is used twice to compute the degree estimates of its two endpoints. 
The second step (lines 5-12) aggregates the local releases to compute noisy statistics $\tilde{T}_{disc}$ and $\tilde{N}$ as defined in \req{eq:degreestat}. This aggregation is post-processing using the DP degree estimates and has no privacy loss. 
Finally, these noisy statistics are used to estimate $\hat{\alpha}_{\textsc{local}}$ in step 3. This results in the following two options. 

\vspace{-0.08in}
\mynonparagraph{Option A: $\hat{\alpha}_{\textsc{local}}$ via Discrete Approximation.}
As shown on line 13, the noisy estimates are plugged into \req{eq:discapproximation} for discrete approximation of the DP estimate $\hat{\alpha}_{\textsc{local}}$. 

\vspace{-0.08in}
\mynonparagraph{Option B: $\hat{\alpha}_{\textsc{local}}$ via Numerical Optimization.}
Numerical optimization is performed using noisy $\tilde{T}_{disc}$ and $\tilde{N}$ based on the same analysis for \req{eq:numanalysislnd}. Our local DP MLE is:

\vspace{-0.06in}
\begin{equation}
\label{eq:localnumoptmle}
\begin{aligned}
\hat{\alpha}_{\textsc{local}}&=\argmax_{\alpha>0} 
 \Bigl[-\alpha\widetilde{\sum \ln d_i}-\tilde{N}\ln Z(\alpha)\Bigr]
\end{aligned}
\end{equation}

which is shown on line 14 in \ralg{alg:mleldpdegreestat}. 

\begin{lemma}
Using the Laplace Mechanism and Sequential Composition~\cite{dpbook}, $\hat{\alpha}_{\textsc{local}}$ computed by \ralg{alg:mleldpdegreestat} is $\varepsilon$-edge differentially private. 
\end{lemma}

\begin{algorithm}[t]
\small
    \caption{DP $\hat{\alpha}_{\textsc{local}}$ via Log Statistic Release}
    \label{alg:mleldplogstat}
    \begin{algorithmic}[1]
        \Require{Graph $G$, Minimum Degree $d_{\min}$, Maximum Degree $d_{\max}$, \newline Privacy Budget $\varepsilon$}
        \Ensure{DP Estimate of $\alpha$}
        \Statex \texttt{/* Step 1: Release local log statistic */}
        \ForEach{$v \in G.V$} 
                \State $\tilde{c_v} \gets \ln(\frac{d_v}{d_{\min}-0.5}) + Lap(\frac{2 \times \ln(\frac{d_{\min}+1}{d_{\min}})}{\varepsilon})$    
                \Comment{$\frac{\varepsilon}{2}$ budget split}
                \State \textsc{Release} $\tilde{c_v}$
        \EndFor
        
        \Statex 
        \Statex \texttt{/* Step 2: Aggregate local statistics */}
        \State $\tilde{T}_{\text{disc}} \gets 0$; $\tilde{N} \gets 0$
        \For{$v \in G.V$} 
        \State $\tilde{c} \gets \textsc{LogStatistic}(v)$
        \If{$\tilde{c} \ge \ln(\frac{d_{\min}}{d_{\min}-0.5}$)}
            \State $\tilde{T}_{\text{disc}} \gets \tilde{T}_{\text{disc}} + \tilde{c}$
            \State $\tilde{N} \gets \tilde{N} + 1$
        \EndIf
        \EndFor
        \Statex
        \Statex \texttt{/* Step 3: Estimate $\alpha$ */}
        \Statex \texttt{/* Option A: Discrete Approximation */}
        \State \quad $\hat{\alpha}_{\textsc{local}} \gets 1 + \frac{\tilde{N}}{\tilde{T}_{\text{disc}}}$
        \Statex \texttt{/* Option B: Numerical Optimization */}
        \State \quad $\hat{\alpha}_{\textsc{local}} \gets \underset{\alpha>0}{\mathrm{argmax}}\, \textsc{LogLikelihood}(\tilde{T}_{\text{disc}}, \tilde{N}, d_{\min}, d_{\max}, \alpha)$ 
        \Statex
        \State \Return $\hat{\alpha}_{\textsc{local}}$
    \end{algorithmic}
\end{algorithm}

\vspace{-0.06in}
\subsection{$\hat{\alpha}_{\textsc{local}}$ via Log Statistic Release}
\ralg{alg:mleldplogstat} shows the LEDP computation for $\hat{\alpha}_{\textsc{local}}$ using Laplace mechanism for log statistic release. In the first step (lines 1-4), each node releases its DP estimate of the log statistic using Laplace noise proportional to global sensitivity $
\ln (
    \frac{
      d_{\min} + 1
    }{
      d_{\min}
    }
  )
 $ with the privacy budget split between two edge endpoints.
The second step (lines 5-12) aggregates local contributions as post-processing to compute noisy statistics $\tilde{T}_{disc}$ and $\tilde{N}$ as defined in \req{eq:logstatfinal}. 
These noisy statistics are used to estimate $\hat{\alpha}_{\textsc{local}}$ in step 3, resulting in following two options. 

\vspace{-0.08in}
\mynonparagraph{Option A: $\hat{\alpha}_{\textsc{local}}$ via Discrete Approximation.}
\req{eq:discapproximation} is used for discrete approximation of the DP estimate $\hat{\alpha}_{\textsc{local}}$ using the noisy estimates (line 13 in \ralg{alg:mleldplogstat}).

\vspace{-0.08in}
\mynonparagraph{Option B: $\hat{\alpha}_{\textsc{local}}$ via Numerical Optimization.}
Numerical optimization is performed (line 14 in \ralg{alg:mleldplogstat}) using noisy $\tilde{T}_{disc}$ and $\tilde{N}$ for local DP MLE defined in \req{eq:localnumoptmle}.

\vspace{-0.02in}
\begin{lemma}
Using the Laplace Mechanism and Sequential Composition~\cite{dpbook}, $\hat{\alpha}_{\textsc{local}}$ computed by \ralg{alg:mleldplogstat} is $\varepsilon$-edge differentially private. 
\end{lemma}

\begin{table}[t]
    \centering
    \caption{Graph datasets \& their power-law scaling parameter $\alpha$.}
    \label{tab:datasummary}
    \vspace{-0.13in}
    \begin{adjustbox}{max width=\columnwidth}
    \begin{tabular}{@{}lrrrr@{}}
    \toprule
    \multirow{2}{*}{Graph} & \multirow{2}{*}{Nodes} & \multirow{2}{*}{Edges} & \multicolumn{2}{c}{\hspace{0.2in} Power-law $\alpha$} \\
    & & & \hspace{0.2in} $d_{\min} = 1$ & $d_{\min} = 3$ \\
    \midrule
    wiki & 7,115 & 414,750 & 1.176 & 1.474 \\
    enron & 36,692 & 367,661 & 1.494 & 1.918 \\
    brightkite & 58,228 & 428,156& 1.551 & 1.982 \\
    ego-twitter & 81,306 & 4,841,488& 1.187 & 1.372 \\
    gplus & 107,614 & 60,989,732& 1.126 & 1.222 \\
    stanford & 281,903 & 2,312,497& 1.459 & 2.218 \\ \hline
    syn-power-0 & 100,000 & 1,477,208 & 2.000 & 2.000 \\
    syn-power-1 & 100,000 & 4,010,327 & 2.500 & 2.500 \\
    syn-power-2 & 100,000 & 997,299 & 3.000 & 3.000 \\
    \bottomrule
    \end{tabular}
    \end{adjustbox}
    \vspace{0.01in}
\end{table}

\begingroup
\setlength{\intextsep}{4pt}
\begin{table}[b]
\vspace{0.05in}
\centering
\small
\captionsetup{aboveskip=2pt,belowskip=2pt}
\caption{Algorithm labels used in the evaluation.}
\label{tab:alglabels}
\begin{adjustbox}{max width=\columnwidth}
\begin{tabular}{l c c c l}
\toprule
Label & Model & Release & Estimator & Reference \\
\midrule
\textbf{\textsc{DA}} & Centralized & -- & Discrete Approx. & \ralg{alg:DEDP} \\
\textbf{\textsc{NO}} & Centralized & -- & Numerical Opt. & \ralg{alg:mledp} \\
\textbf{\textsc{DA/DR}} & Local & Degree & Discrete Approx. & \ralg{alg:mleldpdegreestat} (A) \\
\textbf{\textsc{DA/LR}} & Local & Log-Statistic & Discrete Approx. & \ralg{alg:mleldplogstat} (A) \\
\textbf{\textsc{NO/DR}} & Local & Degree & Numerical Opt. & \ralg{alg:mleldpdegreestat} (B) \\
\textbf{\textsc{NO/LR}} & Local & Log-Statistic & Numerical Opt. & \ralg{alg:mleldplogstat} (B) \\
\bottomrule
\end{tabular}
\end{adjustbox}
\end{table}
\endgroup

\vspace{-0.06in}
\section{Experimental Evaluation}
\label{sec:evaluation}
In this section, we evaluate the accuracy of our \eedp{} $\alpha$ estimation algorithms and answer the following research questions:
\begin{description}[noitemsep,leftmargin=*]
\item[RQ1.] Does adding noise directly to sub-components of $\alpha$ estimator provide better estimates compared to the degree distribution based power-law fitting?
\item[RQ2.] How does the accuracy compare for numerical optimization using noisy estimates instead of directly using the closed form discrete approximation?
\item[RQ3.] Does degree release based approach in local model provide higher accuracy compared to solutions based on local log statistic release?
\item[RQ4.] How does the choice of \(d_{\min}\) affect the accuracy and stability of private \(\alpha\) estimation? 

\end{description}

\vspace{-0.03in}
\mynonparagraph{Algorithms.}
With different combinations for estimation methods (discrete approximation versus numerical optimization) and local statistic release (degree versus log statistic), we evaluate the centralized edge DP and LEDP algorithms listed in \rtab{tab:alglabels}.

We compare against the \eedp{} degree distribution based power-law fitting approach from Hay et al.~\cite{degreedistribution} which is developed for central model. This is called \textbf{\textsc{Base}}.

\vspace{-0.03in}
\mynonparagraph{Datasets.}
We test our algorithms on 9 graph datasets: 6 publicly-available datasets from repository ~\cite{ledpcode} based on SNAP~\cite{snapdatasets} (treated to be undirected) and 3 synthetic datasets. \rtab{tab:datasummary} summarizes the datasets. The power-law scaling parameter $\alpha$ values are mostly below 3 for $d_{\min}$ between 1 and 3; this is consistent with previous observations \cite{clauset2009powerlaws} where $\alpha$ is typically below 3. The synthetic datasets are generated using \textsc{Inc-Powerlaw} generator~\cite{powerlawsyn} that produces a simple random graph conforming with a degree sequence corresponding to the given scaling parameter $\alpha$. 

\vspace{-0.05in}
\mynonparagraph{Methodology.}
We set $\varepsilon$ privacy budget to $1.0$  for our experiments. We also conducted experiments where the $\varepsilon$ value is varied between 0.1 and 5 to study performance across different privacy budgets. We report results for $d_{\min}$ values of 1 and 3; while we also considered with $d_{\min}$ values 5 and 10, the non-private MLE of $\alpha$ was outside of the $[0,\infty]$ range which is semantically invalid. For accuracy metric, we measure the $l_1$ error compared to the non-private $\alpha$ parameter of each dataset. Each experiment was repeated 20 times and we report the mean and standard deviation. 

\begin{table}[t]
\centering
\small
\caption{Summary for centralized model. Mean $l_1$ is averaged over all runs and datasets, max $l_1$ is the worst case, and std. range gives the 
per-dataset standard deviation range over 20 runs. Bold marks the lowest mean and max for each $d_{\min}$.}
\label{tab:centralsummary}
\vspace{-0.1in}
\begin{adjustbox}{max width=\columnwidth}
\begingroup
\renewcommand{\arraystretch}{0.98}
\setlength{\aboverulesep}{0pt}
\setlength{\belowrulesep}{0pt}
\setlength{\cmidrulesep}{0.5pt}
\begin{tabular}{l c c c}
\toprule
Method & Mean $l_1$ (\%) & Max $l_1$ (\%) & Std.\ range \\[-0.1em]
\midrule
\multicolumn{4}{c}{$d_{\min}=1$} \\
\cmidrule(lr){1-4}
\textsc{Base} & 15.88 & 76.70 & 5.249--16.443 \\
\textsc{DA} & 9.57 & 17.69 & 0.00063--0.01459 \\
\textsc{NO} & \textbf{0.0049} & \textbf{0.0989} & 0.00075--0.02512 \\
\midrule
\multicolumn{4}{c}{$d_{\min}=3$} \\
\cmidrule(lr){1-4}
\textsc{Base} & 21.79 & 89.80 & 6.962--16.850 \\
\textsc{DA} & 5.28 & 9.47 & 0.00111--0.02877 \\
\textsc{NO} & \textbf{0.0066} & \textbf{0.1115} & 0.00079--0.03200 \\
\bottomrule
\end{tabular}
\endgroup
\end{adjustbox}
\vspace{-0.05in}
\end{table}

\begin{figure}[t]
\vspace{-0.1in}
    \centering
    {\centering\footnotesize
    \textbf{\textsc{Base}} \textcolor[HTML]{1A80BB}{\rule{0.8em}{0.8em}}\quad
    \textbf{\textsc{DA}} \textcolor[HTML]{E9C716}{\rule{0.8em}{0.8em}}\quad
    \textbf{\textsc{NO}} \textcolor[HTML]{BC272D}{\rule{0.8em}{0.8em}}\par}
    \vspace{0.3em}
    \begin{subfigure}[t]{0.6\linewidth}
        \centering
        \includegraphics[height=0.125\textheight,trim={0 10 8 10},clip]{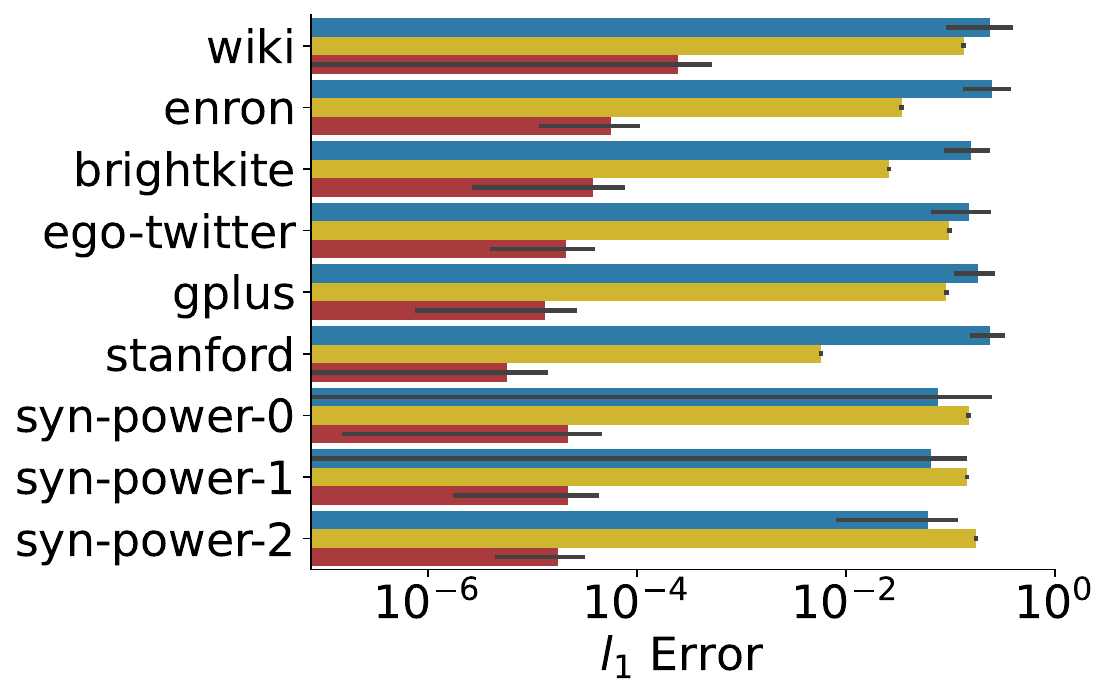}
        \vspace{-0.02in}
        \caption{$d_{\min}=1$}
        \label{fig:centralkmin1}
    \end{subfigure}%
    \hfill%
    \begin{subfigure}[t]{0.4\linewidth}
        \centering
        \includegraphics[height=0.125\textheight,trim={0 8 0 10},clip]{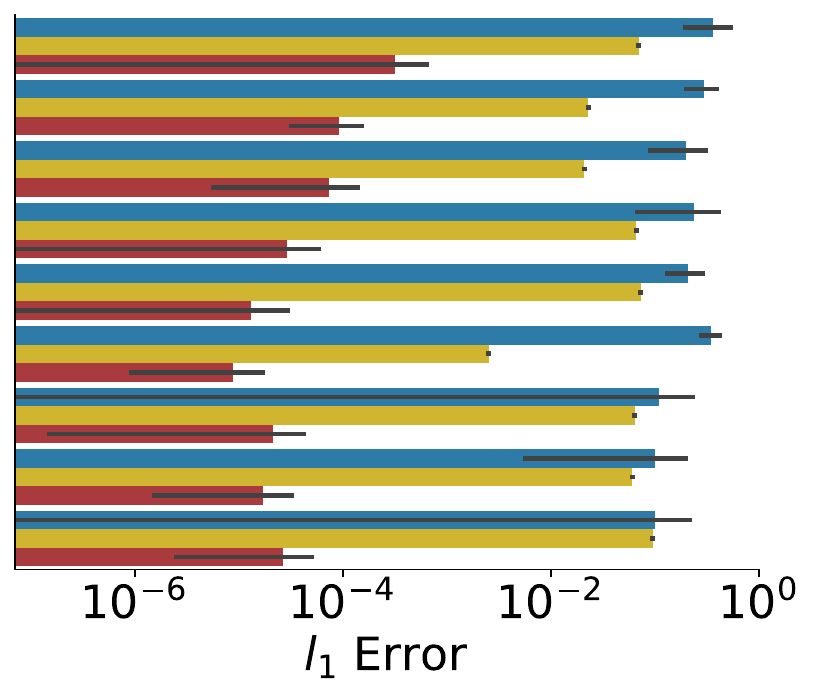}
        \vspace{-0.17in}
        \caption{$d_{\min}=3$}
        \label{fig:centralkmin3}
    \end{subfigure}
    \vspace{-0.12in}
    \caption{Performance of centralized DP algorithms. $l_1$ errors of \textsc{NO}, \textsc{DA}, and \textsc{Base}.}
    \label{fig:centralresults}
\end{figure}

\vspace{-0.04in}
\subsection{Centralized Algorithms}
To answer RQ1 for $\hat{\alpha}_{\textsc{central}}$ estimates, we
compare \textsc{DA} and \textsc{NO} with \textsc{Base}.
\rfig{fig:centralresults} and \rtab{tab:centralsummary} summarize the results. 
Detailed results are available in \rtab{tab:centralperdataset} in \rapp{sec:detailedresults}.

\textsc{NO} is the strongest central method. Its mean
$l_1$ error is roughly three orders of magnitude lower than \textsc{DA}
at both $d_{\min}=1$ and $d_{\min}=3$, and its worst-case error is
about two orders of magnitude lower. This advantage holds on every
individual dataset (Table~\ref{tab:centralperdataset}).
On average, \textsc{DA} is more accurate than \textsc{Base}, at both $d_{\min}$ values, though
\textsc{Base} outperforms \textsc{DA} on the three synthetic power-law datasets at $d_{\min}=1$.
The gap between \textsc{NO} and \textsc{DA} is because \textsc{DA} uses a
closed-form approximation of $\alpha$, while \textsc{NO} optimizes the
exact discrete log-likelihood.

Increasing $d_{\min}$ worsens performance for
\textsc{Base} on every dataset and for \textsc{NO} on most datasets;
\textsc{DA}, in contrast, improves at $d_{\min}=3$ on every dataset.

\subsection{Local Algorithms}
To answer RQ2 and RQ3 for $\hat{\alpha}_{\textsc{local}}$ estimates,
we compare \textsc{NO/LR}, \textsc{DA/LR}, \textsc{NO/DR} and
\textsc{DA/DR}. 
\rfig{fig:localresults} and \rtab{tab:localsummary} summarize the results. 
Detailed results are available in \rtab{tab:localperdataset} in \rapp{sec:detailedresults}.

\textsc{NO/DR} is the strongest local variant overall. Mean $l_1$
error follows \textsc{NO/DR} < \textsc{NO/LR} < \textsc{DA/LR} <
\textsc{DA/DR} at both $d_{\min}=1$ and $d_{\min}=3$
(Table~\ref{tab:localsummary}), and \textsc{NO/DR} has the lowest
per-dataset error on 6 of 9 datasets at each $d_{\min}$ (Table~\ref{tab:localperdataset}).
Release mode interacts with the
estimator: \textsc{NO} prefers degree release, while \textsc{DA} prefers log-statistic release, on most datasets at both $d_{\min}$ values.
All four local variants have lower worst-case error than
\textsc{Base}.
Increasing $d_{\min}$ helps all local variants except
\textsc{NO/DR}; at $d_{\min}=3$, \textsc{DA/LR} improves on 8 of 9
datasets, \textsc{NO/LR} on 7 of 9, and \textsc{DA/DR} on 6 of 9,
while \textsc{NO/DR} worsens on 7 of 9 (Table~\ref{tab:localperdataset}).

\begin{figure}[t]
    \centering
    {\centering\footnotesize
    \textbf{\textsc{NO/LR}} \textcolor[HTML]{A559AA}{\rule{0.8em}{0.8em}}\quad
    \textbf{\textsc{DA/LR}} \textcolor[HTML]{8CCEE3}{\rule{0.8em}{0.8em}}\quad
    \textbf{\textsc{NO/DR}} \textcolor[HTML]{F55F74}{\rule{0.8em}{0.8em}}\quad
    \textbf{\textsc{DA/DR}} \textcolor[HTML]{0D7D87}{\rule{0.8em}{0.8em}}\par}
    \vspace{0.3em}
    \begin{subfigure}[t]{0.6\linewidth}
        \centering
        \includegraphics[height=0.125\textheight,trim={0 10 8 10},clip]{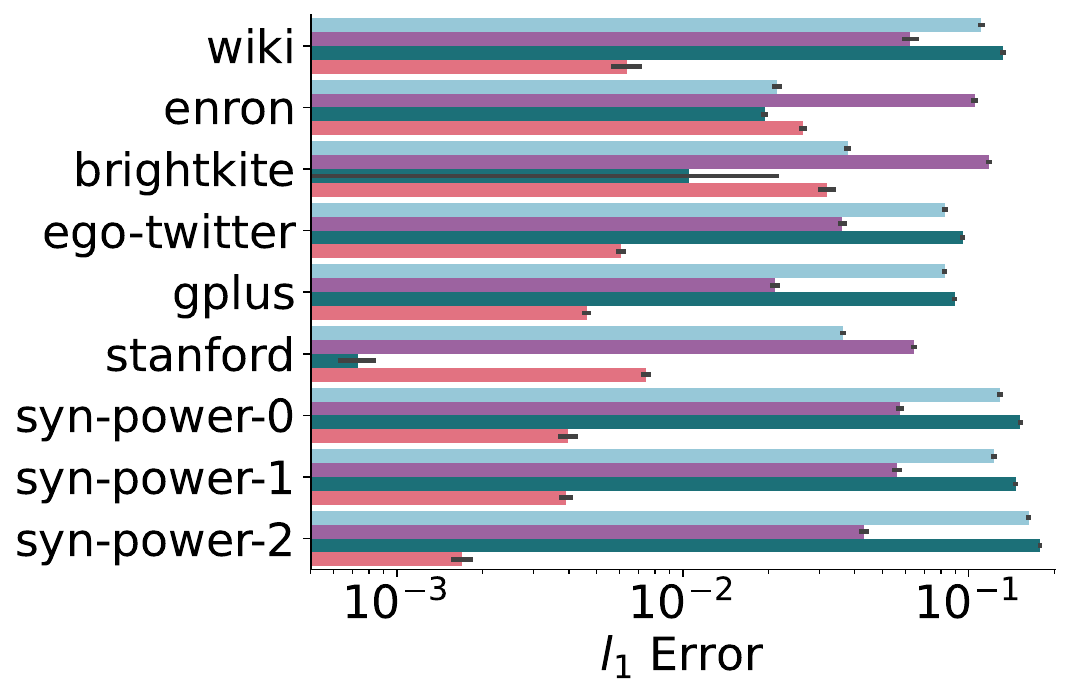}
        \vspace{-0.05in}
        \caption{$d_{\min}=1$}
        \label{fig:localkmin1}
    \end{subfigure}%
    \hfill%
    \begin{subfigure}[t]{0.4\linewidth}
        \centering
        \includegraphics[height=0.125\textheight,trim={0 8 0 10},clip]{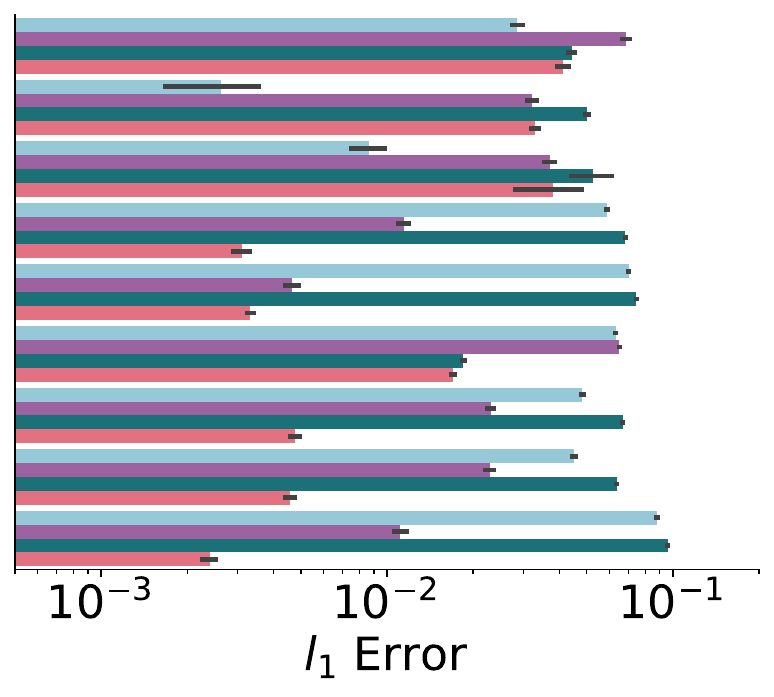}
        \vspace{-0.05in}
        \caption{$d_{\min}=3$}
        \label{fig:localkmin3}
    \end{subfigure}
    \vspace{-0.1in}
    \caption{Performance of local DP algorithms. $l_1$ errors across different combinations for MLE (discrete approximation versus numerical optimization) and local statistic release (degree versus log statistic).}
    \label{fig:localresults}
    \vspace{-0.12in}
\end{figure}

\begin{figure*}[t]
\centering
\captionsetup[subfigure]{skip=2pt}
\captionsetup{aboveskip=3pt}
\begin{minipage}{0.78\textwidth}
\centering\footnotesize
\textbf{Centralized: \ }
\textbf{\textsc{NO}} \textcolor[HTML]{BC272D}{\rule{0.8em}{0.8em}} \quad
\textbf{\textsc{DA}} \textcolor[HTML]{E9C716}{\rule{0.8em}{0.8em}} \qquad \qquad
\textbf{Local: \ }
\textbf{\textsc{NO/LR}} \textcolor[HTML]{A559AA}{\rule{0.8em}{0.8em}} \quad
\textbf{\textsc{DA/LR}} \textcolor[HTML]{8CCEE3}{\rule{0.8em}{0.8em}} \quad
\textbf{\textsc{NO/DR}} \textcolor[HTML]{F55F74}{\rule{0.8em}{0.8em}} \quad
\textbf{\textsc{DA/DR}} \textcolor[HTML]{0D7D87}{\rule{0.8em}{0.8em}} \par
\end{minipage}
\vspace{0.04in}

\begin{minipage}[t]{0.32\textwidth}
    \centering
    \begin{subfigure}[t]{0.485\linewidth}
        \centering
        \includegraphics[width=\linewidth,trim={0 5 5 0},clip]{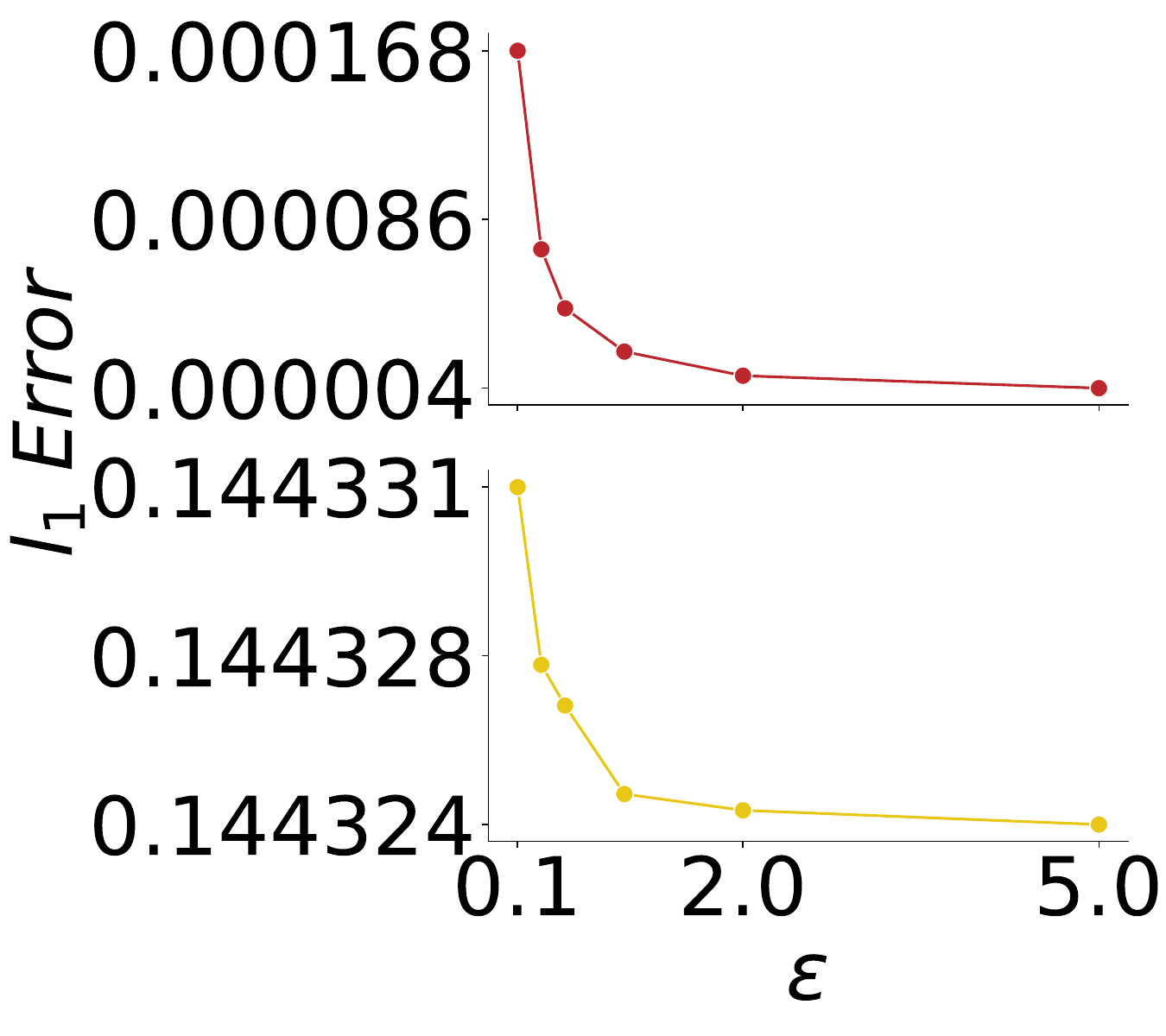}
        \caption{Centralized.}
        \label{fig:centralsynepsilon}
    \end{subfigure}%
    \hfill%
    \begin{subfigure}[t]{0.485\linewidth}
        \centering
        \includegraphics[width=\linewidth,trim={0 5 0 0},clip]{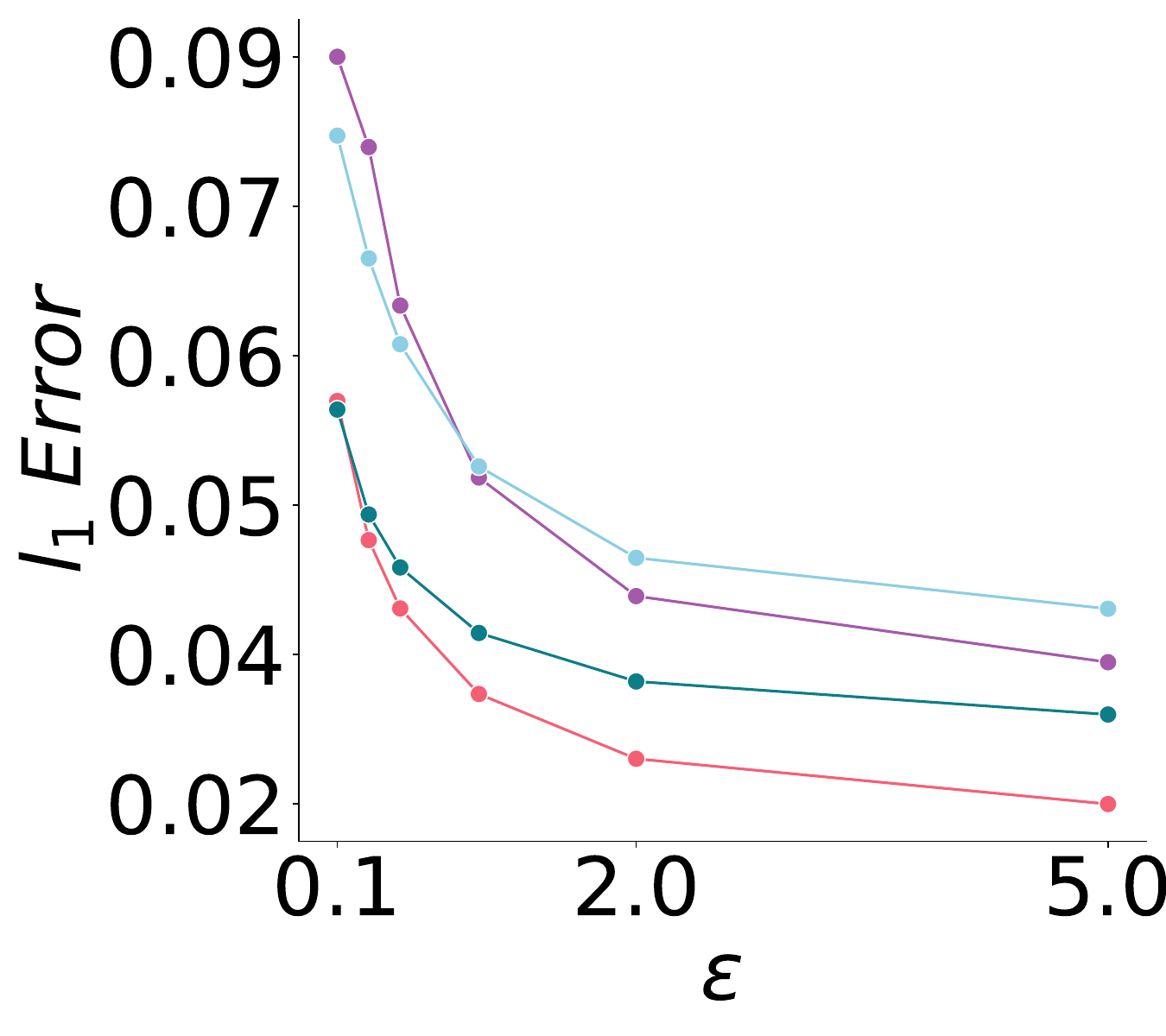}
        \caption{Local.}
        \label{fig:localsynepsilon}
    \end{subfigure}
    \addtocounter{figure}{-1}
    \captionof{figure}{Syn-power-1: $l_1$ vs $\varepsilon$.}
    \label{fig:varyingsynpower}
\end{minipage}%
\hfill
\stepcounter{figure}
\begin{minipage}[t]{0.32\textwidth}
    \centering
    \begin{subfigure}[t]{0.485\linewidth}
        \centering
        \includegraphics[width=\linewidth,trim={0 5 5 0},clip]{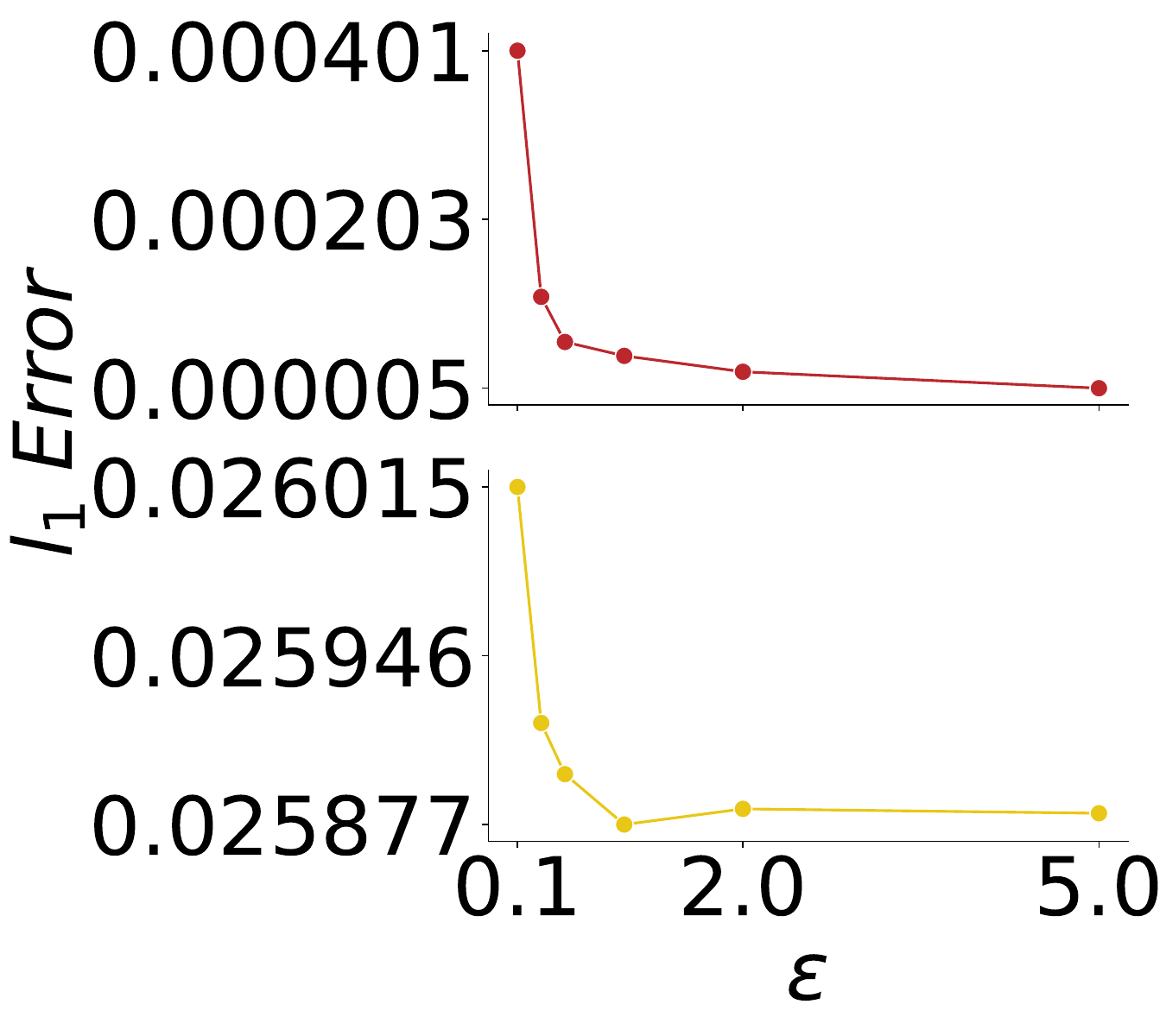}
        \caption{Centralized.}
        \label{fig:centralbrightkite}
    \end{subfigure}%
    \hfill%
    \begin{subfigure}[t]{0.485\linewidth}
        \centering
        \includegraphics[width=\linewidth,trim={0 5 0 0},clip]{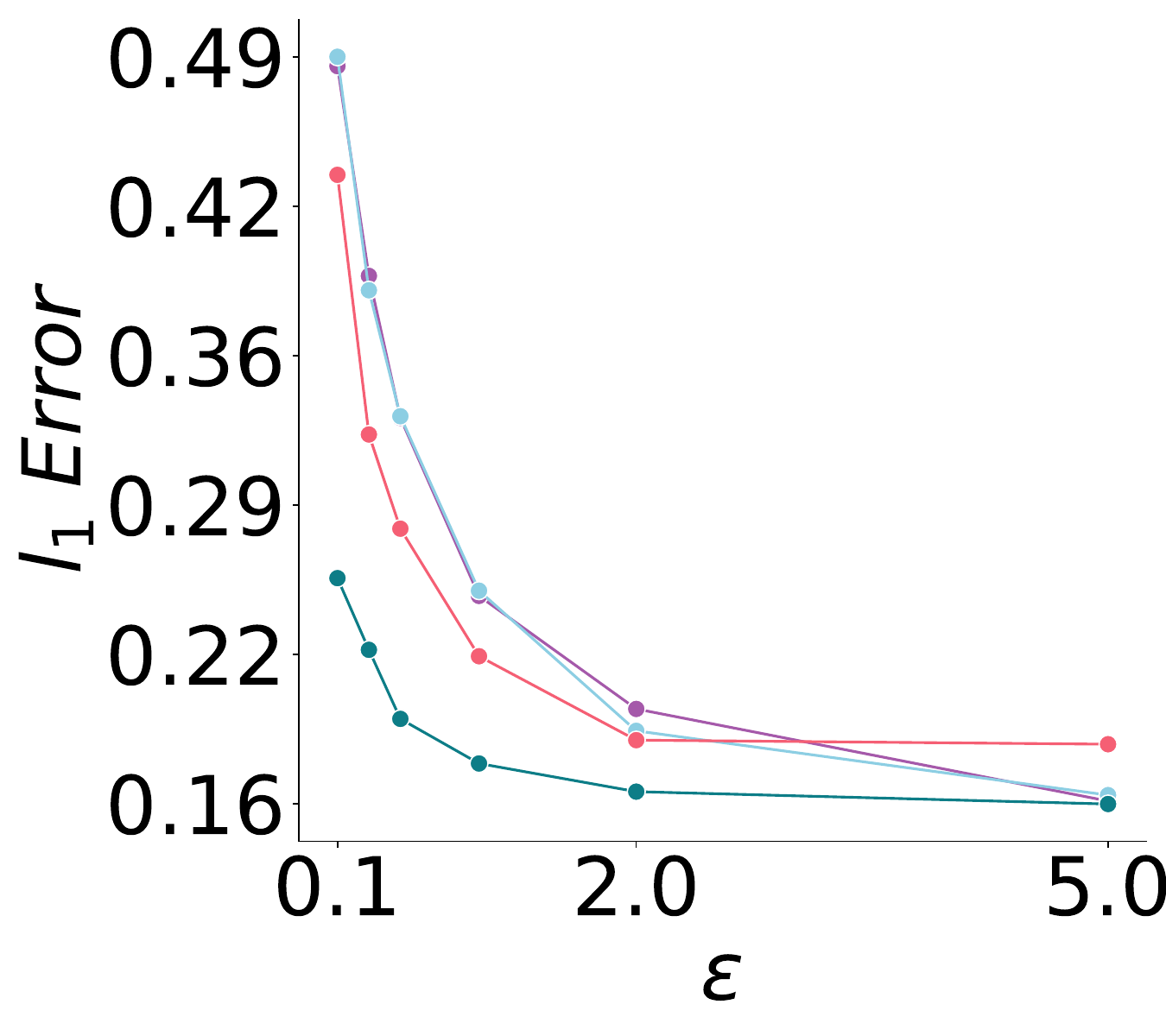}
        \caption{Local.}
        \label{fig:localbrightkite}
    \end{subfigure}
    \addtocounter{figure}{-1}
    \captionof{figure}{Brightkite: $l_1$ vs $\varepsilon$.}
    \label{fig:varyingbrightkite}
\end{minipage}%
\hfill
\stepcounter{figure}
\begin{minipage}[t]{0.32\textwidth}
    \centering
    \begin{subfigure}[t]{0.485\linewidth}
        \centering
        \includegraphics[width=\linewidth,trim={0 5 5 0},clip]{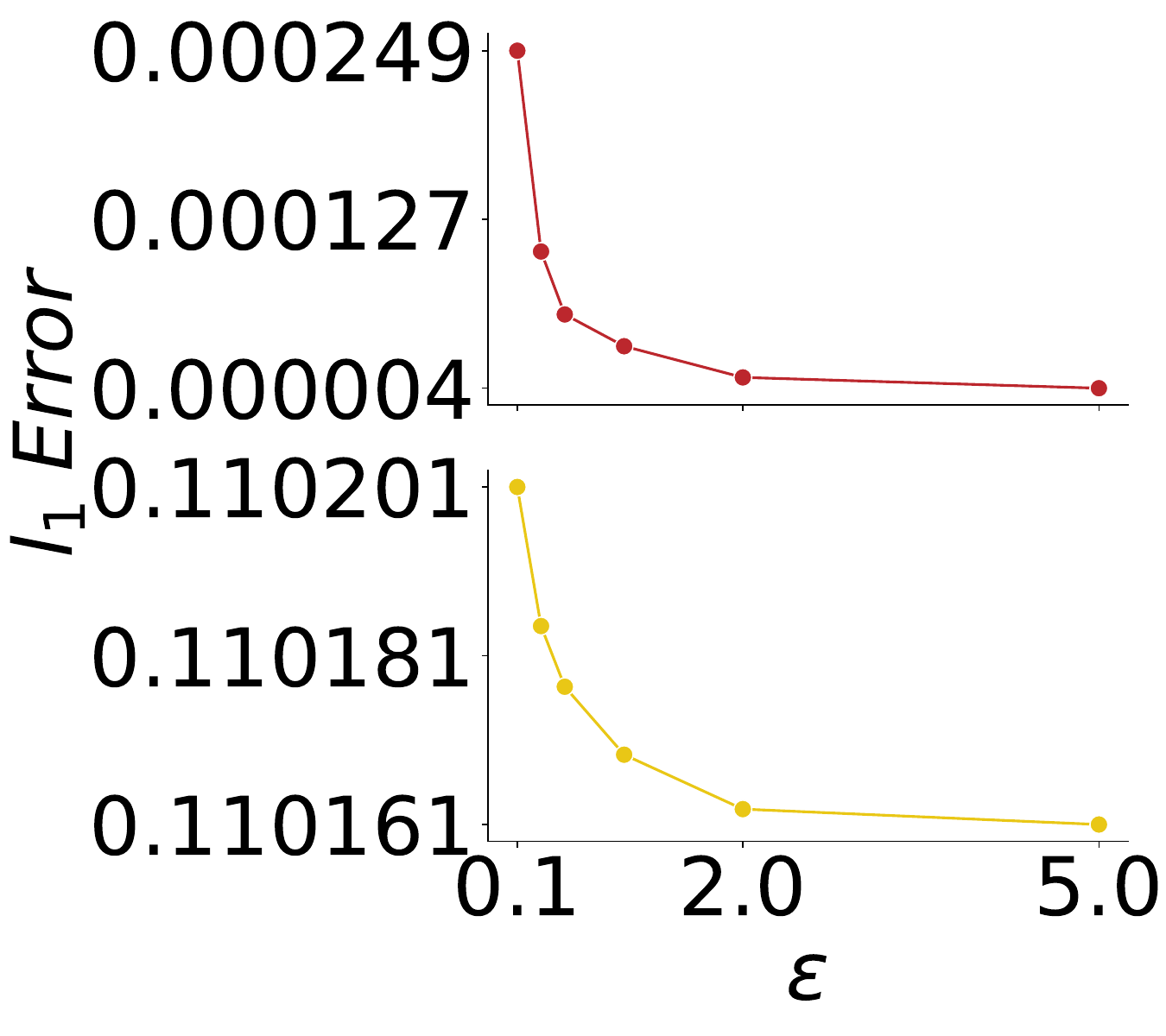}
        \caption{Centralized.}
        \label{fig:centralegotwitter}
    \end{subfigure}%
    \hfill%
    \begin{subfigure}[t]{0.485\linewidth}
        \centering
        \includegraphics[width=\linewidth,trim={0 5 0 10},clip]{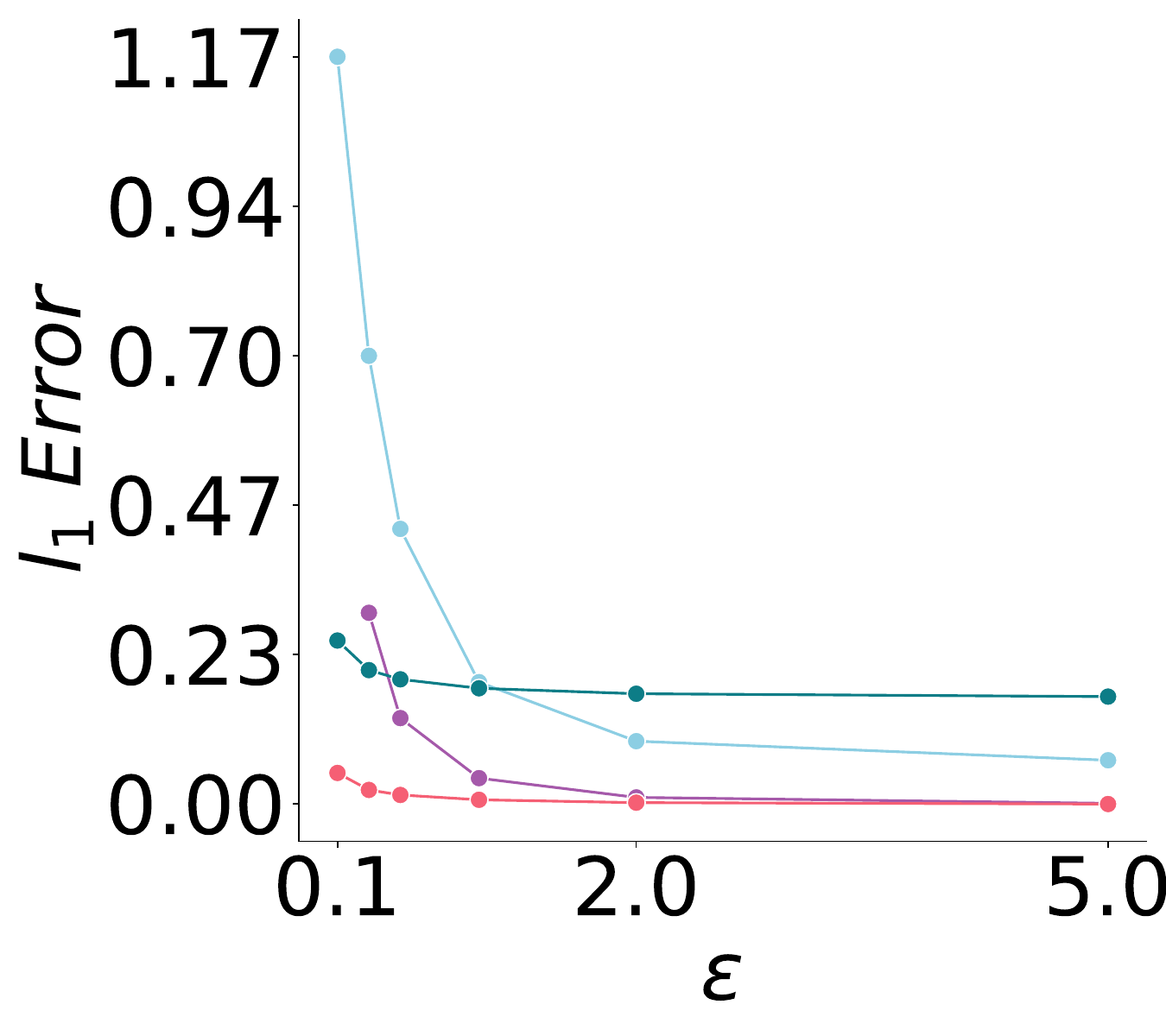}
        \caption{Local.}
        \label{fig:localegotwitter}
    \end{subfigure}
    \addtocounter{figure}{-1}
    \captionof{figure}{Ego-twitter: $l_1$ vs $\varepsilon$.}
    \label{fig:varyingegotwitter}
\end{minipage}
\vspace{-0.12in}
\end{figure*}

\subsection{Sensitivity to Privacy Budget}
We analyze how $l_1$ error changes with privacy budget $\varepsilon$ at
fixed $d_{\min}=1$. \rfig{fig:varyingsynpower},~\rfig{fig:varyingbrightkite}, and~\rfig{fig:varyingegotwitter} show results for three datasets.
Detailed $l_1$ values are available in ~\rtab{tab:epssummary} in \rapp{sec:detailedresults}.

\vspace{-0.02in}
\mynonparagraph{Syn-power-1.}
Figure~\ref{fig:centralsynepsilon} (central, $d_{\min}=1$) shows a nearly flat trend for mean $l_1$ error across varying $\varepsilon$, 
with both
central variants changing very little as privacy budget increases.
In this dataset, \textsc{NO}
has far lower error than \textsc{DA} at every $\varepsilon$. \textsc{DA}
stays nearly same throughout while \textsc{NO} drops as $\varepsilon$ increases. 

Figure~\ref{fig:localsynepsilon} (local, $d_{\min}=1$) shows all
four local curves decreasing as $\varepsilon$ increases. 
Within the log-statistic release family,
\textsc{DA/LR} is better at low privacy budgets ($\varepsilon \le 0.5$),
but \textsc{NO/LR} becomes better for $\varepsilon \ge 1$. Within the
degree-release family, \textsc{DA/DR} is slightly better at
$\varepsilon=0.1$, while \textsc{NO/DR} has lower error for $\varepsilon \ge 0.3$.
Overall, degree-release variants (\textsc{NO/DR}, \textsc{DA/DR})
remain below log-statistic-release variants (\textsc{NO/LR},
\textsc{DA/LR}), and 
the best local accuracy is
obtained by \textsc{NO/DR}.

\setlength{\textfloatsep}{16pt}
\begin{table}[t]
\centering
\small
\captionsetup{aboveskip=0pt,belowskip=2pt}
\caption{Summary for local model. Mean $l_1$ is averaged over all runs and datasets, max $l_1$ is the worst case, and std.\ range gives the per-dataset standard deviation range over 20 runs. Bold marks the lowest mean and max for each $d_{\min}$.}
\label{tab:localsummary}
\vspace{0.1in}
\begin{adjustbox}{max width=\columnwidth}
\begingroup
\renewcommand{\arraystretch}{0.98}
\setlength{\aboverulesep}{0pt}
\setlength{\belowrulesep}{0pt}
\setlength{\cmidrulesep}{0.5pt}
\begin{tabular}{l c c c}
\toprule
Method & Mean $l_1$ (\%) & Max $l_1$ (\%) & Std.\ range \\[-0.1em]
\midrule
\multicolumn{4}{c}{$d_{\min}=1$} \\
\cmidrule(lr){1-4}
\textsc{DA/LR} & 8.72 & 16.27 & 0.01644--0.10478 \\
\textsc{DA/DR} & 9.14 & 17.76 & 0.00295--1.07710 \\
\textsc{NO/LR} & 6.26 & 11.93 & 0.03661--0.30503 \\
\textsc{NO/DR} & \textbf{1.03} & \textbf{3.31} & 0.00817--0.18210 \\
\midrule
\multicolumn{4}{c}{$d_{\min}=3$} \\
\cmidrule(lr){1-4}
\textsc{DA/LR} & 4.61 & 8.88 & 0.01388--0.12399 \\
\textsc{DA/DR} & 5.94 & 9.62 & 0.00528--0.85213 \\
\textsc{NO/LR} & 3.06 & \textbf{7.21} & 0.02496--0.22257 \\
\textsc{NO/DR} & \textbf{1.63} & 7.91 & 0.00927--0.99561 \\
\bottomrule
\end{tabular}
\endgroup
\end{adjustbox}
\end{table}

Comparing across the two models, central variants are less sensitive to
$\varepsilon$, while local variants benefit more from larger $\varepsilon$. 
In absolute error, central \textsc{NO} is below all local variants across the plotted range.

\mynonparagraph{Brightkite.}
Figure~\ref{fig:centralbrightkite} (central, $d_{\min}=1$) shows a
nearly flat trend as $\varepsilon$ increases from 0.1 to 5, 
with small
variation across the full range. 
In this dataset, \textsc{NO} has far lower error than \textsc{DA} at
every $\varepsilon$.

Figure~\ref{fig:localbrightkite} (local, $d_{\min}=1$) shows that
local error decreases overall as $\varepsilon$ increases.
The direct estimation helps the
degree-release family across the full $\varepsilon$ range, while in
the log-statistic family the two estimators are close and trade places
across $\varepsilon$, overall, 
\textsc{DA/DR} is the lowest error local method at every $\varepsilon$
value shown. 
At smaller $\varepsilon$, both degree-release variants (\textsc{NO/DR}, \textsc{DA/DR}) sit below the log-statistic release variants (\textsc{NO/LR}, \textsc{DA/LR}), this advantage narrows with larger $\varepsilon$ and \textsc{NO/DR} rises above the log-statistic curves at $\varepsilon=5$. \textsc{DA/DR}, however, remains the best local option throughout.

Finally, comparing across the two models, the same pattern as \rfig{fig:varyingsynpower} holds:
central curves are less sensitive to $\varepsilon$, while local curves
improve as $\varepsilon$ increases. 
In
absolute error, central variants remain below all local variants across the full range.

\mynonparagraph{Ego-twitter.}
Figure~\ref{fig:centralegotwitter} (central, $d_{\min}=1$) shows a similar trend as previous datasets. 
Figure~\ref{fig:localegotwitter} (local, $d_{\min}=1$) shows a 
stronger dependence on $\varepsilon$, with all local curves
decreasing as privacy budget increases. Within the degree-release
family, \textsc{NO/DR} has lower error than \textsc{DA/DR} at every
$\varepsilon$. Within the log-statistic release family, \textsc{NO/LR}
has no valid estimate at $\varepsilon=0.1$ because all runs hit the
$\hat{\alpha}<0$ clamp, but for $\varepsilon \ge 0.3$ it is
consistently lower error than \textsc{DA/LR}. Overall, \textsc{NO/DR}
is the lowest-error local method at every $\varepsilon$ value shown,
and by $\varepsilon \ge 1$ both numerical-optimization variants
outperform their discrete-approximation counterparts.

Finally, comparing across two models, the same pattern as previous datasets holds:
central curves are less sensitive to $\varepsilon$, while local curves
improve overall as $\varepsilon$ increases. In absolute error, central
\textsc{NO} remains below all reported local estimates across the entire
range.

\subsection{Summary of Findings}
Across all datasets, privacy budgets, and both privacy models, we
highlight four overall observations.

\mynonparagraph{Observation 1: Direct Sub-Component Privatization 
Outperforms Degree-Distribution Fitting.}
To answer RQ1, our results show that adding noise directly to the $   
\alpha$-estimator sub-components is consistently better than degree-distribution based 
fitting. In the centralized setting, \textsc{NO} achieves lower error and better stability than \textsc{Base} on every dataset,
and \textsc{DA} does so in aggregate mean and on all real-world networks.
The main reason is that our method perturbs only low-dimensional sufficient statistics $(T_{\text{disc}},N)$, while degree-distribution based fitting injects noise into many 
histogram bins and can distort the tail used to estimate $\alpha$.

\mynonparagraph{Observation 2: Numerical Optimization is More Accurate Than Direct Approximation Under DP Noise.}
To answer RQ2, \textsc{NO} has lower $l_1$ error than \textsc{DA} in
both models. In the centralized model, \textsc{NO}'s mean $l_1$  error is more than two orders of magnitude lower than that of \textsc{DA} on every dataset, and nearly three orders of magnitude lower on aggregate (Table~\ref{tab:centralperdataset}).

In the local model, \textsc{NO/DR}
is far lower than \textsc{DA/DR}, and \textsc{NO/LR} is lower than
\textsc{DA/LR} on average (Table~\ref{tab:localsummary}). The difference in accuracy is because \textsc{DA}
uses a closed-form approximation of $\alpha$, while \textsc{NO}
optimizes the exact discrete log-likelihood.

\mynonparagraph{Observation 3: In Local DP, Degree Release Helps Numerical Optimization but Log-Statistic Release Helps Direct Approximation.}
To answer RQ3, the effect of release mode depends on the estimator
family. \textsc{NO/DR} has lower mean $l_1$ error than \textsc{NO/LR}
on most datasets, while \textsc{DA/LR} has lower mean $l_1$ error than
\textsc{DA/DR} on most datasets (Table~\ref{tab:localperdataset}).

\mynonparagraph{Observation 4: Effect of $d_{\min}$ is Method-Dependent.}
To answer RQ4, the effect of increasing $d_{\min}$ from
$1$ to $3$ is method-dependent. In the centralized setting,
\textsc{DA} improves on every dataset, while \textsc{NO} worsens on
most (Table~\ref{tab:centralperdataset}). In the local setting, \textsc{DA/LR}, \textsc{DA/DR}, and
\textsc{NO/LR} improve on most datasets, while \textsc{NO/DR} worsens
on most (Table~\ref{tab:localperdataset}). The preferred $d_{\min}$ therefore depends on which estimator
is used.

\section{Related Work}
\vspace{-0.06in}
\mynonparagraph{Differentially Private Degree Distribution Release.}
A common approach for privately estimating $\alpha$ is to fit a power-law model to the privatized degree distribution histogram. 
Hay et al.~\cite{degreedistribution} propose an efficient edge-DP method for releasing the degree distribution and show the fitting power-law estimate. Similarly, works like~\cite{degreecorrelationgraphgen, ASGLDP, directedgraphsdd, node2edge} develop edge-DP degree distribution techniques. DP degree distribution release has also been studied under node-DP \cite{nodedpdegreedistribution,LiuCMY22,nodedpERestimate,  nodehistogram, node2edge}. 
Compared to DP degree distribution use for $\alpha$ estimation, we approach the problem by adding noise only to the few statistics used for $\alpha$ estimation, hence avoiding inaccuracies from binning/smoothing/projection for noisy degree histograms.

\mynonparagraph{Differentially Private Graph Algorithms.}
Beyond private degree distribution, several DP graph algorithms have been developed. They broadly fall into two styles covering both centralized and local models. The first style involves publishing private version of the graph (or a synthetic graph)~\cite{privgraph, ldpsynth, PrivDPR, LDPgen2Phase, ASGLDP, hiddenmarkovgraphgen}, whereas the second style is releasing DP estimations for specific graph queries~\cite{edgedporiginal,ldpstats,shuntriangles,Dhulipala22,imola2, nodedp, nodedprecursivejoins, node2edge}. 
Our work relates follows the second style where we aim to release only the private estimate of a single parameter $\alpha$.  

\vspace{-0.05in}
\mynonparagraph{Differentially Private Likelihood Optimization.}
Many likelihood-based estimators can be written as an optimization problem, \eg, $\hat{\theta}=\arg\max_{\theta}\,\ell(\theta;D)$. 
There are ways to make such estimators private in the central DP model. 
One approach is to privatize the optimization itself, for example by perturbing the objective or by adding noise to the optimizer before release~\cite{ChaudhuriMS11}.
Another way is the McSherry-Talwar selection mechanism $E_q^\varepsilon$~\cite{mcsherry2007mechanism} which is used to select one output from candidates with quality scores (\eg, log-likelihood). 
$E_q^\varepsilon$ samples probabilistically, favoring higher-quality candidates, and providing differential privacy while selecting near-optimal outputs.
More generally, DP M-Estimators can be computed by noisy iterative methods (\eg, adding noise to gradients or Newton steps)
that approximately solve the same estimation problem~\cite{AvellaMedinaBL23}.
These DP ideas are relevant here because $\alpha$ is defined through a likelihood maximization problem, and our approach is to privately release only the few summary numbers for the estimation. 

\vspace{0.04in}
\section{Conclusion and Future Work}
We proposed methods to estimate the power-law exponent $\alpha$ of a graph’s degree distribution under edge differential privacy. 
Our approach privatizes only the small set of sufficient statistics needed for the estimation, aiming to reduce tail distortion and error. We developed both centralized and local edge-DP algorithms with discrete-approximation and numerical-optimization variants, and our experiments showed our direct sufficient-statistic privatization approach is more accurate and stable than histogram-based fitting.

\vspace{0.04in}
\noindent
Interesting directions for future work remain, as discussed next.

\vspace{-0.05in}
\mynonparagraph{Randomized-Response Tail Counts.}
Randomized response~\cite{Kairouz2016} is a local DP method for estimating counts of binary attributes.
In our setting, it could be used to estimate the tail size \mbox{$N=\sum_v \mathbf{1}\{d_v \ge d_{\min}\}$}
by having each node privatize and send only the 1-bit tail-membership indicator, rather than sending
noisy degrees and estimating $N$ by thresholding the noisy values.

\vspace{-0.05in}
\mynonparagraph{Friendship-Paradox Sampling under LDP.}
Although our current estimators assume one privatized report per node, an open direction is to study settings with partial participation or limited communication. Friendship-paradox sampling ~\cite{Nettasinghe2021} provides an alternative way to obtain degree observations that over-represent high-degree nodes, integrating this with LDP would require designing a protocol that collects the necessary neighbor-sampled degree information privately and understanding the privacy/utility tradeoff.

\section*{Acknowledgements}
This work is supported by the National Cybersecurity Consortium and the Natural Sciences and Engineering Research Council of Canada.

\balance
\bibliographystyle{plain}
\bibliography{bib}

@inproceedings{privgraph,
  title={{PrivGraph: Differentially Private Graph Data Publication by Exploiting Community Information}},
  author={Yuan, Quan and Zhang, Zhikun and Du, Linkang and Chen, Min and Cheng, Peng and Sun, Mingyang},
  booktitle={32nd USENIX Security Symposium (USENIX Security 23)},
  pages={3241--3258},
  year={2023}
}

@INPROCEEDINGS{degreedistribution,
  author={Hay, Michael and Li, Chao and Miklau, Gerome and Jensen, David},
  booktitle={Ninth IEEE International Conference on Data Mining}, 
  title={{Accurate Estimation of the Degree Distribution of Private Networks}}, 
  year={2009},
  volume={},
  number={},
  pages={169-178},
  doi={10.1109/ICDM.2009.11}}

@InProceedings{nodedp,
author="Kasiviswanathan, Shiva Prasad
and Nissim, Kobbi
and Raskhodnikova, Sofya
and Smith, Adam",
title={{Analyzing Graphs with Node Differential Privacy}},
booktitle="Theory of Cryptography",
year="2013",
address="Berlin, Heidelberg",
pages="457--476",
isbn="978-3-642-36594-2"
}

@article{edgedporiginal,
author = {Karwa, Vishesh and Raskhodnikova, Sofya and Smith, Adam and Yaroslavtsev, Grigory},
title = {{Private Analysis of Graph Structure}},
year = {2011},
issue_date = {August 2011},
publisher = {VLDB Endowment},
volume = {4},
number = {11},
issn = {2150-8097},
url = {https://doi.org/10.14778/3402707.3402749},
doi = {10.14778/3402707.3402749},
journal = {Proceedings of the VLDB Endowment},
month = aug,
pages = {1146–1157},
numpages = {12}
}

@article{degreecorrelationgraphgen,
author = {Wang, Yue and Wu, Xintao},
title = {{Preserving Differential Privacy in Degree-Correlation based Graph Generation}},
year = {2013},
issue_date = {August 2013},
publisher = {IIIA-CSIC},
address = {Bellaterra, Catalonia, ESP},
volume = {6},
number = {2},
issn = {1888-5063},
journal = {Transactions on Data Privacy},
month = aug,
pages = {127–145},
numpages = {19}
}

@inproceedings {ldpstats,
author = {Jacob Imola and Takao Murakami and Kamalika Chaudhuri},
title = {{Locally Differentially Private Analysis of Graph Statistics}},
booktitle = {30th USENIX Security Symposium (USENIX Security 21)},
year = {2021},
isbn = {978-1-939133-24-3},
pages = {983--1000},
url = {https://www.usenix.org/conference/usenixsecurity21/presentation/imola},
publisher = {USENIX Association},
month = aug
}

@inproceedings{nodedpdegreedistribution,
author = {Day, Wei-Yen and Li, Ninghui and Lyu, Min},
title = {{Publishing Graph Degree Distribution with Node Differential Privacy}},
year = {2016},
isbn = {9781450335317},
url = {https://doi.org/10.1145/2882903.2926745},
doi = {10.1145/2882903.2926745},
booktitle = {{Proceedings of the International Conference on Management of Data}},
pages = {123–138},
numpages = {16},
location = {San Francisco, California, USA},
series = {SIGMOD '16}
}

@inproceedings{ldpsynth,
author = {Qin, Zhan and Yu, Ting and Yang, Yin and Khalil, Issa and Xiao, Xiaokui and Ren, Kui},
title = {{Generating Synthetic Decentralized Social Graphs with Local Differential Privacy}},
year = {2017},
isbn = {9781450349468},
url = {https://doi.org/10.1145/3133956.3134086},
doi = {10.1145/3133956.3134086},
booktitle = {Proceedings of the ACM SIGSAC Conference on Computer and Communications Security},
pages = {425–438},
numpages = {14},
location = {Dallas, Texas, USA},
series = {CCS '17}
}

@article{dpbook,
author = {Dwork, Cynthia and Roth, Aaron},
title = {{The Algorithmic Foundations of Differential Privacy}},
year = {2014},
issue_date = {Aug 2014},
publisher = {Now Publishers Inc.},
address = {Hanover, MA, USA},
volume = {9},
number = {3–4},
issn = {1551-305X},
url = {https://doi.org/10.1561/0400000042},
doi = {10.1561/0400000042},
journal = {Foundations and Trends in Theoretical Computer Science},
month = aug,
pages = {211–407},
numpages = {197}
}

@misc{snapdatasets,
  author       = {Jure Leskovec and Andrej Krevl},
  title        = {{SNAP Datasets: Stanford Large Network Dataset Collection}},
  howpublished = {\url{http://snap.stanford.edu/data}},
  month        = jun,
  year         = 2014
}

@article{centrallocaledp,
  title={{Triangle Counting with Local Edge Differential Privacy}},
  author={Eden, Talya and Liu, Quanquan C and Raskhodnikova, Sofya and Smith, Adam},
  journal={Random Structures \& Algorithms},
  volume={66},
  number={4},
  pages={e70002},
  year={2025},
  publisher={Wiley Online Library}
}

@inproceedings {imola2,
author = {Jacob Imola and Takao Murakami and Kamalika Chaudhuri},
title = {{Communication-Efficient Triangle Counting under Local Differential Privacy}},
booktitle = {31st USENIX Security Symposium (USENIX Security 22)},
year = {2022},
isbn = {978-1-939133-31-1},
address = {Boston, MA},
pages = {537--554},
url = {https://www.usenix.org/conference/usenixsecurity22/presentation/imola},
publisher = {USENIX Association},
month = aug
}

@inproceedings{LiuCMY22,
  author       = {Shang Liu and
                  Yang Cao and
                  Takao Murakami and
                  Masatoshi Yoshikawa},
  title        = {{A Crypto-Assisted Approach for Publishing Graph Statistics with Node Local Differential Privacy}},
  booktitle    = {{IEEE} International Conference on Big Data (Big Data), Osaka, Japan, December 17--20, 2022},
  pages        = {5765--5774},
  year         = {2022},
  doi          = {10.1109/BIGDATA55660.2022.10020435}
}

@article{Newman2005Power,
  author    = {Newman, M. E. J.},
  title     = {{Power Laws, Pareto Distributions and Zipf's Law}},
  journal   = {Contemporary Physics},
  volume    = {46},
  number    = {5},
  pages     = {323--351},
  year      = {2005},
  doi       = {10.1080/00107510500052444},
  publisher = {Taylor \& Francis}
}

@article{ChaudhuriMS11,
  author  = {Kamalika Chaudhuri and Claire Monteleoni and Anand D. Sarwate},
  title   = {{Differentially Private Empirical Risk Minimization}},
  journal = {Journal of Machine Learning Research},
  year    = {2011},
  volume  = {12},
  number  = {29},
  pages   = {1069--1109},
  url     = {http://jmlr.org/papers/v12/chaudhuri11a.html}
}

@inproceedings{mcsherry2007mechanism,
author = {McSherry, Frank and Talwar, Kunal},
title = {{Mechanism Design via Differential Privacy}},
booktitle = {Annual IEEE Symposium on Foundations of Computer Science (FOCS)},
year = {2007},
month = {October},
publisher = {IEEE},
url = {https://www.microsoft.com/en-us/research/publication/mechanism-design-via-differential-privacy/},
edition = {Annual IEEE Symposium on Foundations of Computer Science (FOCS)},
}

@article{AvellaMedinaBL23,
author = {Marco Avella-Medina and Casey Bradshaw and Po-Ling Loh},
title = {{Differentially Private Inference via Noisy Optimization}},
volume = {51},
journal = {The Annals of Statistics},
number = {5},
publisher = {Institute of Mathematical Statistics},
pages = {2067 -- 2092},
year = {2023},
doi = {10.1214/23-AOS2321},
URL = {https://doi.org/10.1214/23-AOS2321}
}

@article{mitzenmacher2004brief,
  title={{A Brief History of Generative Models for Power Law and Lognormal Distributions}},
  author={Mitzenmacher, Michael},
  journal={Internet mathematics},
  volume={1},
  number={2},
  pages={226--251},
  year={2004},
  publisher={Taylor \& Francis}
}

@INPROCEEDINGS{Dhulipala22,
  author={Dhulipala, Laxman and Liu, Quanquan C. and Raskhodnikova, Sofya and Shi, Jessica and Shun, Julian and Yu, Shangdi},
  booktitle={IEEE 63rd Annual Symposium on Foundations of Computer Science (FOCS)}, 
  title={{Differential Privacy from Locally Adjustable Graph Algorithms: k-Core Decomposition, Low Out-Degree Ordering, and Densest Subgraphs}}, 
  year={2022},
  volume={},
  number={},
  pages={754-765},
  doi={10.1109/FOCS54457.2022.00077}}

@article{shuntriangles,
  title={{Practical and Accurate Local Edge Differentially Private Graph Algorithms}},
  author={Mundra, Pranay and Papamanthou, Charalampos and Shun, Julian and Liu, Quanquan C},
  journal={Proceedings of the VLDB Endowment},
  volume={18},
  number={11},
  pages={4199--4213},
  year={2025},
  publisher={VLDB Endowment}
}

@article{li2017experimental,
  title={{An Experimental Study on Hub Labeling Based Shortest Path Algorithms}},
  author={Li, Ye and U, Leong Hou and Yiu, Man Lung and Kou, Ngai Meng},
  journal={Proceedings of the VLDB Endowment},
  volume={11},
  number={4},
  pages={445--457},
  year={2017},
  publisher={VLDB Endowment}
}

@inproceedings{fastexacthub,
author = {Akiba, Takuya and Iwata, Yoichi and Yoshida, Yuichi},
title = {{Fast Exact Shortest-path Distance Queries on Large Networks by Pruned Landmark Labeling}},
year = {2013},
isbn = {9781450320375},
url = {https://doi.org/10.1145/2463676.2465315},
doi = {10.1145/2463676.2465315},
booktitle = {{Proceedings of the ACM SIGMOD International Conference on Management of Data}},
pages = {349–360},
numpages = {12},
location = {New York, New York, USA},
series = {SIGMOD '13}
}

@article{jiang2014hop,
  title={{Hop Doubling Label Indexing for Point-to-Point Distance Querying on Scale-Free Networks}},
  author={Jiang, Minhao and Fu, Ada Wai-Chee and Wong, Raymond Chi-Wing and Xu, Yanyan},
  journal={Proceedings of the VLDB Endowment},
  volume={7},
  number={12},
  year={2014}
}

@article{chung2002connected,
  title={{Connected Components in Random Graphs with Given Expected Degree Sequences}},
  author={Chung, Fan and Lu, Linyuan},
  journal={Annals of combinatorics},
  volume={6},
  number={2},
  pages={125--145},
  year={2002},
  publisher={Springer}
}

@inproceedings{tang2015optimizing,
  title={{Optimizing and Auto-tuning Scale-free Sparse Matrix-vector Multiplication on Intel Xeon Phi}},
  author={Tang, Wai Teng and Zhao, Ruizhe and Lu, Mian and Liang, Yun and Huyng, Huynh Phung and Li, Xibai and Goh, Rick Siow Mong},
  booktitle={IEEE/ACM International Symposium on Code Generation and Optimization (CGO)},
  pages={136--145},
  year={2015},
  organization={IEEE}
}

@inproceedings{vora2019lumos,
  title={{Lumos: Dependency-Driven Disk-based Graph Processing}},
  author={Vora, Keval},
  booktitle={USENIX Annual Technical Conference (USENIX ATC 19)},
  pages={429--442},
  year={2019}
}

@inproceedings{powerlawsyn,
  title={{Engineering Uniform Sampling of Graphs with a Prescribed Power-law Degree Sequence}},
  author={Allendorf, Daniel and Meyer, Ulrich and Penschuck, Manuel and Tran, Hung and Wormald, Nick},
  booktitle={2022 Proceedings of the Symposium on Algorithm Engineering and Experiments (ALENEX)},
  pages={27--40},
  year={2022},
  organization={SIAM}
}

@article{sectric,
author = {Xu, Minze and Xie, Zhentai and Wang, Zhibin and Wang, Guangzhan and Lai, Longbin and Zhang, Yuan and Tian, Chen and Zhong, Sheng},
title = {{Sectric: Towards Accurate, Privacy-Preserving and Efficient Triangle Counting}},
year = {2025},
issue_date = {June 2025},
publisher = {VLDB Endowment},
volume = {18},
number = {10},
issn = {2150-8097},
url = {https://doi.org/10.14778/3748191.3748202},
doi = {10.14778/3748191.3748202},
journal = {Proceedings of the VLDB Endowment},
month = jun,
pages = {3382–3395},
numpages = {14}
}

@article{privagm,
author = {Wang, Songlei and Zheng, Yifeng and Jia, Xiaohua and Hu, Haibo},
title = {{PrivAGM: Secure Construction of Differentially Private Directed Attributed Graph Models on Decentralized Social Graphs}},
year = {2025},
issue_date = {July 2025},
publisher = {VLDB Endowment},
volume = {18},
number = {11},
issn = {2150-8097},
url = {https://doi.org/10.14778/3749646.3749722},
doi = {10.14778/3749646.3749722},
journal = {Proceedings of the VLDB Endowment},
month = jul,
pages = {4682–4694},
numpages = {13}
}

@ARTICLE{ASGLDP,
  author={Wei, Chengkun and Ji, Shouling and Liu, Changchang and Chen, Wenzhi and Wang, Ting},
  journal={IEEE Transactions on Information Forensics and Security}, 
  title={{AsgLDP: Collecting and Generating Decentralized Attributed Graphs With Local Differential Privacy}}, 
  year={2020},
  volume={15},
  number={},
  pages={3239-3254},
  doi={10.1109/TIFS.2020.2985524}}

@inproceedings{nodedpERestimate,
 author = {Ullman, Jonathan and Sealfon, Adam},
 booktitle = {{Advances in Neural Information Processing Systems}},
 editor = {H. Wallach and H. Larochelle and A. Beygelzimer and F. d\textquotesingle Alch\'{e}-Buc and E. Fox and R. Garnett},
 pages = {},
 publisher = {Curran Associates, Inc.},
 title = {{Efficiently Estimating Erdos-Renyi Graphs with Node Differential Privacy}},
 url = {https://proceedings.neurips.cc/paper_files/paper/2019/file/955cb567b6e38f4c6b3f28cc857fc38c-Paper.pdf},
 volume = {32},
 year = {2019}
}

@incollection{directedgraphsdd,
author = "Reuben, Jenni",
title = {{Towards a Differential Privacy Theory for Edge-Labeled Directed Graphs}},
year = 2018,
doi = "10.18420/sicherheit2018_24",
booktitle = "SICHERHEIT 2018",
publisher = "Gesellschaft für Informatik e.V.",
address = "Bonn",
pissn = "1617-5468",
isbn = "978-3-88579-675-6",
pages = "273--278",
}

@inproceedings{nodedprecursivejoins,
author = {Chen, Shixi and Zhou, Shuigeng},
title = {{Recursive Mechanism: Towards Node Differential Privacy and Unrestricted Joins}},
year = {2013},
isbn = {9781450320375},
url = {https://doi.org/10.1145/2463676.2465304},
doi = {10.1145/2463676.2465304},
booktitle = {{Proceedings of the ACM SIGMOD International Conference on Management of Data}},
pages = {653–664},
numpages = {12},
location = {New York, New York, USA},
series = {SIGMOD '13}
}

@ARTICLE{nodehistogram,
  author={Liu, Ganghong and Ma, Xuebin and Li, Wuyungerile},
  journal={IEEE Access}, 
  title={{Publishing Node Strength Distribution With Node Differential Privacy}}, 
  year={2020},
  volume={8},
  number={},
  pages={217642-217650},
  doi={10.1109/ACCESS.2020.3040077}}

@article{node2edge,
author = {Hu, Yihua and Ding, Hao and Dong, Wei},
title = {{N2E: A General Framework to Reduce Node-Differential Privacy to Edge-Differential Privacy for Graph Analytics}},
year = {2025},
issue_date = {December 2025},
volume = {3},
number = {6},
url = {https://doi.org/10.1145/3769808},
doi = {10.1145/3769808},
journal = {Proceedings of the ACM on Management of Data},
month = dec,
articleno = {343},
numpages = {26}
}

@article{clauset2009powerlaws,
   title={{Power-Law Distributions in Empirical Data}},
   volume={51},
   ISSN={1095-7200},
   url={http://dx.doi.org/10.1137/070710111},
   DOI={10.1137/070710111},
   number={4},
   journal={SIAM Review},
   publisher={Society for Industrial & Applied Mathematics (SIAM)},
   author={Clauset, Aaron and Shalizi, Cosma Rohilla and Newman, Mark EJ},
   year={2009},
   month=nov, pages={661–703} }

@article{bauke2007powerlaws,
   title={{Parameter Estimation for Power-law Distributions by Maximum Likelihood Methods}},
   volume={58},
   ISSN={1434-6036},
   url={http://dx.doi.org/10.1140/epjb/e2007-00219-y},
   DOI={10.1140/epjb/e2007-00219-y},
   number={2},
   journal={The European Physical Journal B},
   publisher={Springer Science and Business Media LLC},
   author={Heiko Bauke},
   year={2007},
   month=jul, pages={167–173} }

@article{Nettasinghe2021,
author = {Nettasinghe, Buddhika and Krishnamurthy, Vikram},
title = {{Maximum Likelihood Estimation of Power-law Degree Distributions via Friendship Paradox-based Sampling}},
year = {2021},
issue_date = {June 2021},
volume = {15},
number = {6},
issn = {1556-4681},
url = {https://doi.org/10.1145/3451166},
doi = {10.1145/3451166},
journal = {ACM Transactions on Knowledge Discovery from Data (TKDD)},
month = may,
articleno = {106},
numpages = {28}
}

@article{Kairouz2016,
  title={{Extremal Mechanisms for Local Differential Privacy}},
  author={Kairouz, Peter and Oh, Sewoong and Viswanath, Pramod},
  journal={Advances in Neural Information Processing Systems},
  volume={27},
  year={2014}
}

@inproceedings{PrivDPR,
author = {Zhang, Sen and Hu, Haibo and Ye, Qingqing and Xu, Jianliang},
title = {{PrivDPR: Synthetic Graph Publishing with Deep PageRank under Differential Privacy}},
year = {2025},
isbn = {9798400712456},
url = {https://doi.org/10.1145/3690624.3709334},
doi = {10.1145/3690624.3709334},
booktitle = {Proceedings of the 31st ACM SIGKDD Conference on Knowledge Discovery and Data Mining V.1},
pages = {1936–1947},
numpages = {12},
location = {Toronto ON, Canada},
series = {KDD '25}
}

@inproceedings{LDPgen2Phase,
author = {Zhang, Yuxuan and Wei, Jianghong and Zhang, Xiaojian and Hu, Xuexian and Liu, Wenfen},
title = {{A Two-Phase Algorithm for Generating Synthetic Graph Under Local Differential Privacy}},
year = {2018},
isbn = {9781450365673},
url = {https://doi.org/10.1145/3290480.3290503},
doi = {10.1145/3290480.3290503},
booktitle = {Proceedings of the 8th International Conference on Communication and Network Security},
pages = {84–89},
numpages = {6},
location = {Qingdao, China},
series = {ICCNS '18}
}

@ARTICLE{hiddenmarkovgraphgen,
  journal={IEEE Transactions on Computers}, 
  title={{Dynamic Graph Publication With Differential Privacy Guarantees for Decentralized Applications}}, 
  author={Li, Zhetao and Xiao, Yong and Liu, Haolin and Liao, Xiaofei and Yuan, Ye and Du, Junzhao},
  year={2025},
  volume={74},
  number={5},
  pages={1771-1785},
  doi={10.1109/TC.2025.3543605}
}

@software{ledpcode,
  author       = {Pranay Mundra and
                  Quanquan C. Liu},
  title        = {{mundrapranay/DistributedLEDPGraphAlgos: VLDB
                   Artifact Release, Zenodo
                  }},
  month        = jun,
  year         = 2025,
  publisher    = {Zenodo},
  version      = {v1.0},
  doi          = {10.5281/zenodo.15741880},
  url          = {https://doi.org/10.5281/zenodo.15741880},
  swhid        = {swh:1:dir:582e42ebbd113ac598fd3ce00940ffbb1a9350bf
                   ;origin=https://doi.org/10.5281/zenodo.15741879;vi
                   sit=swh:1:snp:b2cd2ae2859debd273eee9836f14770f9e13
                   b24b;anchor=swh:1:rel:830b61123a35ce7096ed97467ae6
                   184c5c55b0e4;path=mundrapranay-
                   DistributedLEDPGraphAlgos-fc4d2fe
                  },
}

\newpage
\appendix
\section{Detailed Results}
\label{sec:detailedresults}

This appendix reports the detailed per-dataset and per-$\varepsilon$
results that support the main experimental comparisons in the paper.

\begin{table}[H]
\centering
\small
\caption{Mean $l_1$ (\%) in centralized model, averaged over 20 runs. Bold marks the lowest mean within each each dataset and $d_{\min}$ block.}
\label{tab:centralperdataset}
\vspace{-0.1in}
\begin{adjustbox}{max width=\columnwidth}
\begingroup
\renewcommand{\arraystretch}{0.98}
\setlength{\aboverulesep}{0pt}
\setlength{\belowrulesep}{0pt}
\setlength{\cmidrulesep}{0.5pt}
\begin{tabular}{l ccc ccc}
\toprule
\multirow{2}{*}{Dataset} & \multicolumn{3}{c}{$d_{\min}=1$} & \multicolumn{3}{c}{$d_{\min}=3$} \\[-0.1em]
\cmidrule(lr){2-4}\cmidrule(lr){5-7}
& \textsc{Base} & \textsc{DA} & \textsc{NO} & \textsc{Base} & \textsc{DA} & \textsc{NO} \\
\midrule
wiki & 23.90 & 13.44 & \textbf{0.0248} & 36.25 & 6.97 & \textbf{0.0315} \\
enron & 25.14 & 3.43 & \textbf{0.0057} & 29.45 & 2.29 & \textbf{0.0091} \\
brightkite & 15.86 & 2.59 & \textbf{0.0038} & 20.00 & 2.09 & \textbf{0.0073} \\
ego-twitter & 15.15 & 9.82 & \textbf{0.0021} & 23.66 & 6.61 & \textbf{0.0029} \\
gplus & 18.53 & 9.15 & \textbf{0.0013} & 20.98 & 7.26 & \textbf{0.0013} \\
stanford & 24.14 & 0.58 & \textbf{0.0006} & 34.60 & 0.25 & \textbf{0.0009} \\
syn-power-0 & 7.55 & 15.02 & \textbf{0.0022} & 10.93 & 6.36 & \textbf{0.0021} \\
syn-power-1 & 6.53 & 14.43 & \textbf{0.0022} & 10.10 & 6.05 & \textbf{0.0017} \\
syn-power-2 & 6.10 & 17.69 & \textbf{0.0017} & 10.08 & 9.46 & \textbf{0.0026} \\
\bottomrule
\end{tabular}
\endgroup
\end{adjustbox}
\end{table}

\begin{table}[H]
\centering
\small
\caption{Mean $l_1$ (\%) in local model, averaged over 20 runs. Bold marks the lowest mean within each dataset and $d_{\min}$ block.}
\label{tab:localperdataset}
\vspace{-0.1in}
\begin{adjustbox}{max width=\columnwidth}
\begingroup
\renewcommand{\arraystretch}{0.98}
\setlength{\aboverulesep}{0pt}
\setlength{\belowrulesep}{0pt}
\setlength{\cmidrulesep}{0.5pt}
\begin{tabular}{l c c c c c}
\toprule
Dataset & \textsc{DA/LR} & \textsc{DA/DR} & \textsc{NO/LR} & \textsc{NO/DR} \\[-0.1em]
\midrule
\multicolumn{5}{c}{$d_{\min}=1$} \\
\cmidrule(lr){1-5}
wiki & 11.09 & 13.18 & 6.25 & \textbf{0.64} \\
enron & 2.14 & \textbf{1.94} & 10.49 & 2.63 \\
brightkite & 3.78 & \textbf{1.05} & 11.80 & 3.20 \\
ego-twitter & 8.26 & 9.54 & 3.62 & \textbf{0.61} \\
gplus & 8.26 & 8.94 & 2.10 & \textbf{0.46} \\
stanford & 3.63 & \textbf{0.07} & 6.44 & 0.74 \\
syn-power-0 & 12.86 & 15.18 & 5.76 & \textbf{0.40} \\
syn-power-1 & 12.29 & 14.59 & 5.61 & \textbf{0.39} \\
syn-power-2 & 16.20 & 17.75 & 4.30 & \textbf{0.17} \\
\midrule
\multicolumn{5}{c}{$d_{\min}=3$} \\
\cmidrule(lr){1-5}
wiki & \textbf{2.86} & 4.42 & 6.86 & 4.13 \\
enron & \textbf{0.26} & 5.00 & 3.23 & 3.30 \\
brightkite & \textbf{0.86} & 5.26 & 3.71 & 3.81 \\
ego-twitter & 5.90 & 6.80 & 1.14 & \textbf{0.31} \\
gplus & 7.00 & 7.45 & 0.47 & \textbf{0.33} \\
stanford & 6.30 & 1.85 & 6.50 & \textbf{1.70} \\
syn-power-0 & 4.83 & 6.69 & 2.30 & \textbf{0.48} \\
syn-power-1 & 4.51 & 6.36 & 2.28 & \textbf{0.46} \\
syn-power-2 & 8.81 & 9.61 & 1.11 & \textbf{0.24} \\
\bottomrule
\end{tabular}
\endgroup
\end{adjustbox}
\end{table}

\begin{table}[H]
\centering
\small
\caption{Mean $l_1$ error for $\varepsilon$ with $d_{\min}=1$ for Syn-power-1, Brightkite, and Ego-twitter. Bold marks the best method within the Central and Local blocks for each $\varepsilon$ within each dataset. The Ego-twitter NO/LR entry at $\varepsilon=0.1$ is marked -- because all runs optimized to $\hat{\alpha}<0$; with the minimum $\hat{\alpha}$ clamped to $0$, no valid estimates remained.}
\label{tab:epssummary}
\vspace{-0.13in}
\begin{adjustbox}{max width=\columnwidth}
\begingroup
\renewcommand{\arraystretch}{0.98}
\setlength{\aboverulesep}{0pt}
\setlength{\belowrulesep}{0pt}
\setlength{\cmidrulesep}{0.5pt}
\begin{tabular}{c c c c c c c}
\toprule
\multirow{2}{*}{$\varepsilon$} & \multicolumn{2}{c}{Central} & \multicolumn{4}{c}{Local} \\[-0.1em]
\cmidrule(lr){2-3}\cmidrule(lr){4-7}
& \textsc{DA} & \textsc{NO} & \textsc{DA/LR} & \textsc{DA/DR} & \textsc{NO/LR} & \textsc{NO/DR} \\
\midrule
\multicolumn{7}{c}{Syn-power-1} \\[-0.1em]
\midrule
0.1 & 0.14433 & \textbf{0.00017} & 0.08033 & \textbf{0.05712} & 0.08701 & 0.05784 \\
0.3 & 0.14433 & \textbf{0.00007} & 0.06993 & 0.04822 & 0.07936 & \textbf{0.04605} \\
0.5 & 0.14433 & \textbf{0.00004} & 0.06265 & 0.04373 & 0.06593 & \textbf{0.04025} \\
1.0 & 0.14432 & \textbf{0.00002} & 0.05229 & 0.03818 & 0.05135 & \textbf{0.03301} \\
2.0 & 0.14432 & \textbf{0.00001} & 0.04455 & 0.03407 & 0.04130 & \textbf{0.02752} \\
5.0 & 0.14432 & \textbf{0.00000} & 0.04023 & 0.03127 & 0.03570 & \textbf{0.02368} \\
\midrule
\multicolumn{7}{c}{Brightkite} \\[-0.1em]
\cmidrule(lr){1-7}
0.1 & 0.02601 & \textbf{0.00040} & 0.48725 & \textbf{0.25672} & 0.48313 & 0.43505 \\
0.3 & 0.02592 & \textbf{0.00011} & 0.38401 & \textbf{0.22502} & 0.39038 & 0.32023 \\
0.5 & 0.02590 & \textbf{0.00006} & 0.32831 & \textbf{0.19449} & 0.32760 & 0.27859 \\
1.0 & 0.02588 & \textbf{0.00004} & 0.25118 & \textbf{0.17476} & 0.24886 & 0.22218 \\
2.0 & 0.02588 & \textbf{0.00002} & 0.18919 & \textbf{0.16231} & 0.19887 & 0.18508 \\
5.0 & 0.02588 & \textbf{0.00001} & 0.16078 & \textbf{0.15683} & 0.15816 & 0.18328 \\
\midrule
\multicolumn{7}{c}{Ego-twitter} \\[-0.1em]
\cmidrule(lr){1-7}
0.1 & 0.11020 & \textbf{0.00025} & 1.17006 & 0.25609 & -- & \textbf{0.04877} \\
0.3 & 0.11018 & \textbf{0.00010} & 0.70188 & 0.20976 & 0.29952 & \textbf{0.02217} \\
0.5 & 0.11018 & \textbf{0.00006} & 0.43090 & 0.19537 & 0.13465 & \textbf{0.01429} \\
1.0 & 0.11017 & \textbf{0.00003} & 0.19093 & 0.18131 & 0.04065 & \textbf{0.00675} \\
2.0 & 0.11016 & \textbf{0.00001} & 0.09852 & 0.17279 & 0.01045 & \textbf{0.00225} \\
5.0 & 0.11016 & \textbf{0.00000} & 0.06864 & 0.16830 & 0.00118 & \textbf{0.00010} \\
\bottomrule
\end{tabular}
\endgroup
\end{adjustbox}
\end{table}

\end{document}